\documentclass[journal]{IEEEtran}
\usepackage{verbatim}
\usepackage{graphicx}
\usepackage{algorithm,algorithmic}
\usepackage{amsmath,amssymb,amsfonts,bm}
\usepackage{subfigure}
\usepackage{psfrag}
\usepackage[dvips]{epsfig}
\usepackage{amsthm}
\usepackage{color}
\usepackage[usenames,dvipsnames]{pstricks}
\usepackage[dvips]{epsfig}
\usepackage{pst-grad} % For gradients
\usepackage{pst-plot} % For axes
\usepackage{tikz,pgf}
\usepackage{xpatch,letltxmacro}
\tikzset{font={\fontsize{10pt}{12}\selectfont}}
\newcommand\scalemath[2]{\scalebox{#1}{\mbox{\ensuremath{\displaystyle #2}}}}
\usepackage{amsmath,graphicx}
\usepackage{amsmath,graphicx}
\usepackage{enumerate}
\usepackage{amsbsy}
\usepackage{amssymb}
\usepackage{amsthm}
\usepackage{amscd}
\usepackage{subfigure}
\usepackage{color}
\usepackage{cite}
\usepackage{listings}
\usepackage{tikz}
\tikzset{font={\fontsize{10pt}{12}\selectfont}}
\usepackage{pgf}

\newtheorem{lemma}{Lemma}
\usepackage{amsthm}
\usepackage{stackrel}
\usepackage{pifont}     % For checkmark and xmark
\usepackage{tabularx}   % For table with adjustable column width
\usepackage{booktabs}   % For better horizontal rules
\usepackage{multirow}
\usepackage{makecell}
\newtheorem{rema}{Remark}

\newtheorem{prop}{Proposition}
\newcommand{\cmark}{\ding{51}}  % check mark
\newcommand{\xmark}{\ding{55}}  % x mark

\def\bB{{\bf B}}

\def\bI{{\bf I}}

\def\bA{{\bf A}}

\def\complexC{{\mathbb{C}}}
\def\realR{{\mathbb{R}}}

\def\wi{10}
\def\he{4}
\makeatletter
\newcommand*{\rom}[1]{\expandafter\@slowromancap\romannumeral #1@}
\begin{document}
\title{Multi-UAV Enabled Integrated Sensing and Wireless Powered Communication: A Robust Multi-Objective Approach}

\author{Omid Rezaei, Mohammad~Mahdi~Naghsh\IEEEauthorrefmark{1}, \emph{Senior~Member, IEEE}, Seyed Mohammad Karbasi, \emph{Senior~Member, IEEE}, Mohammad Mahdi Nayebi, \emph{Senior~Member, IEEE}, and Saeed~Gazor, \emph{Senior~Member, IEEE}
	
	\thanks{A limited part of this work has been published in IEEE International Conference on Acoustics, Speech and Signal Processing (ICASSP), Rhodes Island, Greece, June 2023 \cite{10094604}.
		
		O. Rezaei, S. M. Karbasi, and M. M. Nayebi are with the Department of Electrical Engineering, Sharif University of Technology, Tehran, 11155-4363, Iran. M. M. Naghsh is with the Department of Electrical and Computer Engineering, Isfahan University of Technology, Isfahan, 84156-83111, Iran. S. Gazor is with the Department of Electrical and Computer Engineering, Queen’s University, Kingston, Ontario, K7L 3N6, Canada. *Please address all the correspondence to M. M. Naghsh, Phone: (+98) 31-33912450; Fax: (+98) 31-33912451; Email: mm\_naghsh@iut.ac.ir.
		
		Authors acknowledge the support of the Natural Sciences and Engineering Research Council of Canada (NSERC).} }

\maketitle

\begin{abstract}
This paper addresses the optimization of integrated sensing and communication (ISAC) systems in unmanned aerial vehicle (UAV)-aided wireless networks featuring wireless power transfer (WPT). We propose a novel architecture wherein multiple UAV-based radars concurrently serve multiple clusters of energy-limited communication users while performing sensing tasks. Initially, radars sense the environment, allowing communication users to harvest and store energy from radar transmissions. Subsequently, this stored energy facilitates uplink communication from nodes to UAVs. Our multi-objective design problem optimizes UAV trajectories, radar transmit waveforms, radar receive filters, time scheduling, and uplink powers to enhance both radar and communication system performance. Incorporating user location uncertainty, we formulate a robust non-convex optimization problem. To address this, we employ an alternating optimization approach, complemented by fractional programming, S-procedure, and majorization-minimization (MM) techniques. Numerical examples illustrate the efficacy of our method across diverse scenarios.
\end{abstract}
\begin{keywords}
	Integrated sensing and communication (ISAC), multi-objective optimization, unmanned aerial vehicle (UAV), wireless powered communication network (WPCN).
\end{keywords}
\section{Introduction}
In recent years, driven by enormous demand for concurrent sensing and communication, there has been significant interest in developing a new paradigm referred to as integrated sensing and communication (ISAC) in both academia and industry \cite{hassanien2019dual}. Note that research related to this topic has also been addressed with different terminologies like joint radar communication (JRC) \cite{ref10}, joint communication and radar sensing (JCAS) \cite{zhang2018multibeam}, dual-functional radar communication (DFRC) \cite{wang2020constrained}, and radar communication (RadCom) \cite{zhang2020joint}. Thanks to its high mobility and flexibility, unmanned aerial vehicle (UAV) has emerged as a key technology in future ISAC systems \cite{meng2022uav}. Precisely, it is expected that UAVs will bring better coverage and improved sensing and communication services in modern ISAC systems \cite{meng2022uav1}.

On the other hand, the integration of UAVs into wireless powered communication networks (WPCNs) has attracted a significant attention \cite{10098733}. UAV-enabled WPCNs leverage the maneuverability of UAVs to serve as aerial access points (AP)s, facilitating wireless charging of low-power internet-of-thing devices while simultaneously collecting information from them \cite{xie2020common}. Unlike conventional WPCNs with fixed APs, UAV-enabled WPCNs offer rapid deployment capabilities, especially in critical situations like natural disasters. Additionally, the controllable mobility of UAVs enables dynamic adjustment of their locations, optimizing trajectories to enhance wireless coverage and improve both wireless information transfer and wireless power transfer efficiency.

This paper presents a novel approach to UAV-aided integrated sensing and wireless powered communication by unifying the concepts of ISAC and WPCN, with particular emphasis on joint waveform design. Specifically, we propose a scheme in which UAV-based radars transmit sensing waveforms during the wireless power transfer (WPT) phase of the WPCN. These waveforms are meticulously designed to reflect the key characteristics of both radar sensing and WPT, enabling joint optimization of essential radar and power transfer parameters. By utilizing backscattered signals, targets can be detected while energy from the transmitted waveforms is harvested by energy-constrained ground users. In the subsequent wireless information transfer (WIT) phase, users upload their information signals to the UAVs. Our approach addresses the limitations of prior works, which typically assume that UAVs employ information-bearing signals for radar probing—an assumption that diminishes radar performance due to the low energy allocated for sensing and the inherent properties of communication signals. By performing sensing during the WPT phase of the WPCN, we aim to enhance the overall efficiency and effectiveness of integrated sensing and communication systems in UAV-aided wireless networks.
\subsection{Abbreviations}
The abbreviations utilized in this paper are presented in Table~\ref{table21}.
\begin{table}
	\centering
	\caption{List of Abbreviations.}\label{table21}
	\begin{tabular}{|c|c|} 
		\hline
		\textbf{Abbreviation}&\textbf{Term} \\ \hline
		AP& access point\\ \hline
		CoMP&  coordinated multi-point\\ \hline
		CSCG&  circularly symmetric complex Gaussian\\ \hline
		D&  dimensional\\ \hline
		DC&direct current\\ \hline
		DFRC&dual-functional radar communication\\
		\hline 
		EH&energy harvesting\\
		\hline 
		FoV&field-of-view\\ \hline
		GPS&global positioning system\\ \hline
		ISAC&integrated sensing and communication\\ \hline
		JCAS&joint communication and radar sensing\\ \hline
		JRC&joint radar communication\\ \hline
		LoS&line-of-sight\\ \hline
		MM&majorization-minimization\\ \hline
		MRT&maximum ratio transmission \\ \hline
		NFZ&no-fly zone\\ \hline
		PMCW&phase modulated continuous
		wave\\ \hline
		RadCom&radar communication\\ \hline
		RF&radio frequency\\ \hline
		SINR&signal-to-interference-plus-noise ratio\\ \hline
		SWIPT&simultaneous wireless information and power transfer\\ \hline
		UAV&unmanned aerial vehicle\\ \hline
		UPA&uniform planar array\\ \hline
		WIT& wireless information transfer\\ \hline
		WPCN&wireless powered communication network\\ \hline
		WPT&wireless power transfer\\ \hline
		w.r.t.&with respect to\\ \hline
	\end{tabular} 
\end{table}
\subsection{Related Works}
\subsubsection{UAV-aided ISAC}
Over the past decade, UAVs have emerged as effective and economical solutions for diverse applications in wireless communication and sensing domains, leveraging their exceptional mobility, flexibility, and rapid deployment capabilities \cite{10246260}. In wireless communication, UAVs have been instrumental in enabling energy-efficient communication \cite{9847346}, facilitating cooperative relaying \cite{7192644}, and supporting EH-based wireless communication \cite{8941314}. Moreover, UAVs have shown promise in enhancing various sensing applications, including weather radar \cite{8818617}, vital sign detection \cite{9478877}, etc.

Furthermore, there are several works in the literature that have considered the UAV-enabled ISAC \cite{lyu2021joint,meng2022throughput,10107972,chen2020performance,ref8,wang2020constrained,ref9,ref12,ref13}. The authors in \cite{lyu2021joint} have considered a sum-rate maximization problem constrained to sensing requirements for given targets in a single UAV-enabled ISAC model. In contrast with the mentioned work that the UAV must provide the communication services and sensing tasks simultaneously, the authors in \cite{meng2022throughput} considered a periodic sensing and communication strategy for their UAV-enabled ISAC model to separate the practical requirements of sensing and communication over time.

It is possible to realize more effective sensing and communication with multi-UAV cooperation compared to a single UAV with limited sensing coverage and communication capability. In a recent work, the authors in \cite{10107972} proposed a UAV-aided ISAC system integrated with mobile edge computing, aiming to minimize system energy and time costs through joint optimization of UAV trajectory, sensing scheduling, and resource allocation. The work in \cite{chen2020performance} proposed the problem of UAV sensing range maximization based on mutual sensing interference and the communication capacity constraints in a cooperative multi-UAV network. In \cite{ref8}, the completion time minimization problem for multi-UAV ISAC systems is studied. By considering a given required localization accuracy for radar
sensing, the authors in \cite{wang2020constrained} studied maximization of both sum and minimum communication rates under the Cramer-Rao bound constraint of the target localization in a multi-UAV ISAC system. In \cite{ref9}, authors proposed a framework based on the extended Kalman filter to track the ground users in a multi-UAV ISAC network. Then, a UAV swarm-enabled ISAC model in \cite{ref12} considered a distributed cooperative framework for multi-target tracking. In \cite{ref13}, the resource allocation problem for a multi-UAV ISAC system is addressed via a method based on reinforcement learning.
\subsubsection{UAV-aided WPCN}
Some previous studies have explored the potential of UAVs as alternatives for traditional fixed APs in WPCN \cite{ 924852310, 948449611, 808617712}. For instance, in \cite{ 924852310}, researchers investigated throughput maximization in a multi-UAV enabled WPCN operating in the millimeter-wave frequency range. Meanwhile, in \cite{948449611}, a different approach was pursued, with one UAV dedicated to WPT and another to WIT. To optimize the trajectories of these UAVs, a multi-agent deep Q-network framework was proposed. Furthermore, \cite{808617712} addressed the resource allocation challenge in a UAV-aided WPCN to facilitate device-to-device communication, employing the Lagrangian relaxation method for the resource allocation algorithm in the mentioned study.
\begin{table*}
	\centering
	\caption{Comparative Summary of Key Features Across Related Works.}\label{tab:comparison}
	\begin{tabular}{|c|c|c|c|c|c|c|c|c|} 
		\hline
	&\textbf{Multi-UAV}&\textbf{Robust}&\textbf{Sensing in}&\textbf{Waveform}&\textbf{Filter} &\textbf{Trajectory}&\textbf{Multi-Objective}& \textbf{Non-linear}\\
		 &\textbf{Scenario}&\textbf{Design}&\textbf{WPT Phase} &\textbf{Design}&\textbf{Design}&\textbf{Optimization}&\textbf{(Radar + Communication)}&\textbf{EH} 
		 	\\ \hline 
		 	\cite{ref9}, \cite{808617712}& \xmark & \xmark & \xmark & \xmark & \xmark & \xmark & \xmark & \xmark 
		 	\\ \hline
		 	 \cite{lyu2021joint,meng2022throughput,10107972}& \xmark & \xmark & \xmark & \xmark & \xmark & \cmark & \xmark & \xmark 
		 	\\ \hline
		 	\cite{chen2020performance}, \cite{ref12}, \cite{ 924852310}& \cmark & \xmark & \xmark & \xmark & \xmark & \xmark & \xmark & \xmark 
		 	\\ \hline
		 	\cite{ref8}, \cite{948449611}& \cmark & \xmark & \xmark & \xmark & \xmark & \cmark & \xmark & \xmark
		 	\\ \hline
		 	\cite{ref13}& \cmark & \xmark & \xmark & \xmark & \xmark & \xmark & \cmark & \xmark
		 	\\ \hline
		 	\cite{10663809}, \cite{10556683}& \xmark & \xmark & \cmark & \xmark & \xmark & \xmark & \xmark & \xmark
		 	\\ \hline
		 	\cite{10382465}, \cite{10539920}& \xmark & \xmark & \cmark & \xmark & \xmark & \xmark & \cmark & \xmark
		 	\\ \hline
		 	\cite{10622740}& \xmark & \xmark & \cmark & \cmark~(Partially) & \xmark & \xmark & \xmark & \cmark
		 	\\ \hline
		 \cite{10094604}& \xmark & \xmark & \cmark & \cmark~(Partially) & \cmark & \cmark & \cmark & \xmark
		\\ \hline
			This paper& \cmark & \cmark & \cmark & \cmark & \cmark & \cmark & \cmark & \cmark \\
			\hline
	\end{tabular} 
\end{table*}
\begin{table*}
	\centering
	\caption{List of Important Variable Notations Used in This Paper.}\label{table2}
	\begin{tabular}{|c|c|c|c|}
		\hline
		\textbf{Variable}&\textbf{Meaning}&\textbf{Variable}&\textbf{Meaning}\\ \hline
		$M$&Number of UAVs/clusters&$p_{k,m}^{ul}[n]$&Uplink transmission power of $\mathrm{U}_{k,m}$\\ \hline
		$K$&Number of ground users&$\widetilde{\mathbf{w}}_m$&Radar receive filter of $\mathrm{UAV}_m$\\ \hline
		$T$&Total time period&$\widetilde{\mathbf{y}}_{m,l} [n]$&Received echo signal at $\mathrm{UAV}_m$\\ \hline
		$N$&Number of discrete time parts&$\bar{\mathbf{y}}_{m,l} [n]$&Output of range processing block at $\mathrm{UAV}_m$\\ \hline
		$\mathrm{U}_{k,m}$&$k$th user in $m$th cluster&$\widetilde{\mathbf{a}}_{m}$&Doppler processing filter\\ \hline
		$\mathrm{UAV}_m$&$m$th UAV&${{\mathrm{SINR}}}_{m,l}\hspace{1pt}[n]$ &Sensing SINR of $\mathrm{UAV}_m$ after range processing\\ \hline
		$\delta_t$&Duration of each time slot&${\widetilde{\mathrm{SINR}}}_{m}\hspace{1pt}[n]$&Sensing SINR of $\mathrm{UAV}_m$ after Doppler processing\\ \hline
		${\mathbf{q}}_m [n]$&3D coordinate of $\mathrm{UAV}_m$&$\hm{\Xi}_{m,l}$&Noise-plus-interference matrix associated with ${{\mathrm{SINR}}}_{m,l}\hspace{1pt}[n]$\\ \hline
		$\widetilde{\mathbf{q}}_m [n]$&2D coordinate of $\mathrm{UAV}_m$&$\widetilde{\hm{\Xi}}_{m}$&Noise-plus-interference matrix associated with ${\widetilde{\mathrm{SINR}}}_{m}\hspace{1pt}[n]$\\ \hline
		$h_{k,i,m} [n]$&Channel power gain from $\mathrm{UAV}_m$ to $\mathrm{U}_{k,i}$&$E_{k,m} [n]$&Linear harvested energy at $\mathrm{U}_{k,i}$\\ \hline
		$\tau_0$&sensing/WPT time fraction&$E^0_{k,m}$&Remaining energy from previous period at $\mathrm{U}_{k,i}$\\ \hline
		$\tau_{k,m}[n]$&WIT time fraction for $\mathrm{U}_{k,m}$&$E^{nl}_{k,m} [n]$&Non-linear harvested energy at $\mathrm{U}_{k,i}$\\ \hline
		$\widetilde{\mathbf{x}}_m$&PMCW radar sequence for $\mathrm{UAV}_m$&$\widetilde{ E}_{k,m} [n]$&A modified version of $E^{nl}_{k,m} [n]$\\ \hline
		$L$&Integration parameter&\scriptsize $\bar{d}_{k,m}$&Radius of the circular uncertainty region for $\mathrm{U}_{k,m}$\\ \hline
		$\widetilde{\tau}_0$& time fraction for each $\widetilde{\mathbf{x}}_m$&$\widetilde{r}_{k,m}$ &Normalized radius of user location uncertainty for $\mathrm{U}_{k,m}$\\ \hline
		$\widetilde{N}$&Number of sub-sequences of each $\widetilde{\mathbf{x}}_m$&$\mathrm{\widetilde{SIR}}_{m,k}$&Radar signal-to-interference ratio for $k$th index at $\mathrm{UAV}_m$\\ \hline
		$p_m^{dl}[n]$&Downlink transmission power of $\mathrm{UAV}_m$&$\mathrm{\widetilde{SNR}}_{m,l}$&Radar signal-to-noise ratio for $l$th sequence at $\mathrm{UAV}_m$\\ \hline
	\end{tabular} 
\end{table*}
\subsection{Contribution}
The aforementioned works on UAV-aided ISAC have assumed that UAVs employ information-bearing signals as the radar probing signals as well; this leads to limitation of radar capabilities due to availability of low sensing energy as well as special characteristics of communication signals. This observation motivates us to combine the concepts of UAV-aided ISAC and WPCN. In the conference version of the current work \cite{10094604}, we firstly propose an approach in which sensing is conducted during the WPT phase, aiming to overcome the limitations observed in prior research  \cite{ahmed2025advancements,chin2024multi,10892027}. In this extended version, we build upon and complete the ideas introduced in the conference version. Specifically, the main contributions of this paper are summarized as follows.

\subsubsection{Proposed model}  In this paper, we consider a multi-UAV enabled integrated sensing and wireless powered communication, where sensing is done in the WPT phase of the communication. Consequently, the radar/WPT waveforms can be designed for sensing purpose with more power leading to radar performance improvement. Precisely, in the first phase of our model, multiple UAV-based PMCW radars transmit sensing waveforms and then, targets (e.g., non-authorized UAVs \cite{guerra2022networks,guvenc2018detection}) can be detected by filtering the backscattered signals. The energy of these sensing waveforms can also be harvested by the energy-limited ground users, utilizing a non-linear EH model. Then, in the WIT phase, the users in each cluster can upload their information signals to their associated UAVs. 
\subsubsection{Joint waveform design}
Designing a waveform that jointly addresses requirements of sensing and WPT is a challenging task. Unlike prior works in the literature--that largely focus on powering users within ISAC frameworks but overlook the critical role of the sensing waveform\footnote{For example, see \cite{10382465,10663809,10556683} for SWIPT and \cite{10539920, 10622740} for WPT-enabled ISAC.}--we delve into a comprehensive approach for joint sensing and WPT waveform design. Specifically, the sensing waveform in our proposed PMCW radar comprises a series of consecutive PMCW sequences. The duration of these PMCW sequences, along with the number of sequence repetitions, is determined by various factors such as range resolution, maximum unambiguous range, radar receiver sampling frequency, minimum required sensing SINR, system complexity constraints, etc. Consequently, the sensing waveform cannot be of arbitrary duration. In contrast, from the perspective of WPT, a continuous waveform duration offers greater flexibility during the design stage. To address this challenge, we have treated the number of sensing sequence repetitions as a design variable with a constraint on its upper and lower bounds, allowing us to simultaneously fulfill the requirements of both sensing and WPT functionalities.
\subsubsection{Problem formulation}
	The aim is to jointly maximize the minimum radar SINR and minimum throughput of communication users by designing the transmit radar/WPT waveforms, radar receive filters, time scheduling as well as uplink power of users, and UAV trajectories under and user location uncertainty. 
\subsubsection{Proposed solution}
	The cast multi-objective robust design problem is non-convex and hence, hard to solve. Therefore, we first adopt the scalarization technique to rewrite the objective as a specific weighted sum of sensing and communication metrics. Then, we devise a method exploiting the concepts of fractional programming, S-procedure, and tricky MM techniques in order to deal with
	the problem efficiently. Our simulation results show the effectiveness of the proposed method.

For completeness, Table~\ref{tab:comparison} provides a comparative summary of key features in recent related works, highlighting the distinctions and advancements introduced in this paper. 	
\subsection{Organization}
The rest of this paper is organized as follows. The signal and system models are explained in Section \ref{sys}. In Section \ref{sum}, a multi-objective problem is formulated, and an optimization framework is proposed for dealing with the design problem. In Section \ref{nonl}, we explore the impact of the non-linear EH circuits on the proposed model. Section \ref{num} presents numerical examples to illustrate the effectiveness of the proposed method. Finally, conclusions are drawn in Section \ref{con}.
\subsection{Notation}
Bold lowercase (uppercase) letters are used for vectors (matrices).
The notations $\Re \{ \cdot \}$, $\mathbb{E} [\cdot ]$, $\vert \cdot \vert$, $\|\cdot\|_2$, ${(\cdot)^{{T}}}$, $(\cdot)^{{H}}$, $(\cdot)^{{*}}$, and $\mbox{tr} \{ \cdot \}$ indicate the real-part, statistical expectation, absolute value, $l_2$-norm of a vector, transpose, Hermitian, complex conjugate, and trace of a matrix, respectively. The symbols $[\mathbf{A}]_{i,j}$ and $[\mathbf{a}]_i$ are used for element-wise representation of the matrix $\mathbf{A}$ and vector $\mathbf{a}$, respectively. The notations $\nabla f(\cdot)$ and $\nabla^2 f(\cdot)$ indicate the gradient and the
Hessian of the twice-differentiable function $f$, respectively. We denote $\mathcal{CN}(\boldsymbol{\omega},\mathbf{\Sigma})$ as a CSCG distribution with mean $\boldsymbol{\omega}$ and covariance $\mathbf{\Sigma}$. The set $\mathbb{R}$ represents real numbers and $\mathbb{R}^{N}$ and $\complexC^{N}$ are the set of ${N \times 1}$ real and complex vectors, respectively. The set of ${N \times N}$ Hermitian and identity matrices are denoted by $\mathbb{H}^{N \times N}$ and $\bI_{N}$, respectively. The notation $\bA \succeq \bB$ means that $\bA-\bB$ is positive semi-definite. 
\section{System Model} \label{sys}
In this section, we present the signal and system models. We include a list of important variable notations in Table~\ref{table2} for improved readability and consistency. As shown in Fig.~\ref{ht}, we consider a multi-UAV enabled integrated sensing and wireless powered communication network where $M$ UAVs are employed to serve $KM$ single-antenna ground users denoted by $\mathrm{U}_{k,m}, \hspace{1pt}1\leq k \leq K,\hspace{1pt}1\leq m \leq M, $ in $M$ clusters and also act as surveillance radars. The number of clusters $M$ and the dimensions of each cluster are contingent upon various factors including user dispersion, UAV altitude during flight, limitations on flight duration, etc. First, the UAVs transmit energy to users and also perform radar sensing. Then, users in each cluster transmit their information signals to their associated UAVs. Each user has an EH circuit and can store energy for its operation. 
In this paper, we assume that UAVs are equipped with directional single-antenna with half power beam-width of $\zeta$ (in degree) for radar/communication transmitters as well as communication receivers; the radar receivers can have several antennas (i.e., an UPA of antennas) to estimate target direction in addition to its range and Doppler. The
UAVs are assumed to fly at the time-varying coordinate $\mathbf{q}_m (t) = [x_m(t), y_m(t), z_m (t)]^T \in \realR^{3},$ at
time interval $0 \leq t \leq T$. The period $T$ is discretized into $N$ equal time slots where the elemental slot
length $\delta_t = T/N$ is chosen to be sufficiently small such that the location of UAVs are considered as approximately unchanged within each time slot $\delta_t$. Note that given the period $T$, the value of $N$ should exceed a lower bound to ensure the desired level of approximation accuracy while remaining below an upper bound to maintain system simplicity \cite{ 8247211}. Therefore, the value of  $N$ can be appropriately selected to strike a balance between accuracy and complexity, ensuring the rationality of $\delta_t$. As a result, the trajectory of $m$th UAV, denoted by $\mathrm{UAV}_m$, can be approximated by
the sequence $\mathbf{q}_m[n] = [x_m[n], y_m[n], z_m[n]]^T,\hspace{1pt} 1 \leq n \leq N$. The user location information at UAVs, provided e.g., by GPS, may be also imperfect due to radio signal interference \cite{cui2018robust}. Thus, in this paper, we take into account the user location uncertainty for robust resource allocation.
\begin{figure}
	\centering
	\begin{tikzpicture}[even odd rule,rounded corners=2pt,x=12pt,y=12pt,scale=.55,every node/.style={scale=.7}]
		\node[inner sep=0pt] (russell) at (2-14-5,9.8)
		{\includegraphics[width=.1\textwidth]{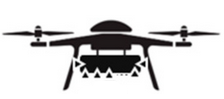}};
		\node [] at (-1-15.5-5,10-.6) {\footnotesize $\hspace{2pt}\mathrm{UAV}_{1}$};
		
		\node [] at (-12+5-2,9-.5) {\LARGE $\hdots$};

		\node[inner sep=0pt] (russell) at (2-9+10,9.8)
		{\includegraphics[width=.1\textwidth]{images2.png}};
		\node [] at (-2+.5,10-.6) {\footnotesize $ \hspace{2pt}\mathrm{UAV}_{M}$};

		\draw [dashed,fill=red!10] (-19.1,-5.4) ellipse(3.2cm and 1.4cm);
		
		\node [] at (-9,-5.4) {\LARGE $\hdots$};
		
		\draw [dashed,fill=red!10] (-19.1+20,-5.4) ellipse(3.2cm and 1.4cm);
		
		\node[inner sep=0pt] (russell) at (-22.5,3)
		{\includegraphics[width=.065\textwidth]{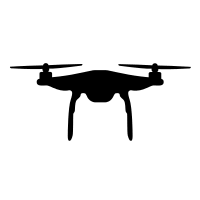}};
	%	\node [] at (-24,2) {\footnotesize ${T}_{{1,1}}$};

		\node [] at (-22.5+3.25,3) {\LARGE $\hdots$};
		
		\node[inner sep=0pt] (russell) at (-22.5+6.5,3)
		{\includegraphics[width=.065\textwidth]{images.png}};
		%\node [] at (-24+6.5,2) {\footnotesize ${T}_{{L,1}}$};

		\node[inner sep=0pt] (russell) at (-11.5+10,4)
		{\includegraphics[width=.065\textwidth]{images.png}};
	%	\node [] at (-13+10,3-.1) {\footnotesize ${T}_{{1,M}}$};

		\node[inner sep=0pt] (russell) at (-3+10,2)
		{\includegraphics[width=.065\textwidth]{images.png}};
	%	\node [] at (-4.5+10,1-.1) {\footnotesize ${T}_{{L,M}}$};

		\draw[thick ,fill=green!10] (-18-6,-6.5) rectangle ++(2.7,2) node[midway]{\footnotesize $\mathrm{U}_{1,1}$};
		\draw[thick] (-17-5+.35-1,-4.5)--+(0,.7);
		\draw[thick,fill=gray!15] (-17-5+.35-1,-3.8)--++(.75,0.5)--++(-1.5,0)--++(.75,-.5)--++(0,-.1);
		
		\draw [rounded corners=0pt,fill=white!10](-19-5-1,-6.4) rectangle ++(.7,1) ;
		\draw [rounded corners=0pt,fill=white!10](-18.8-5-1,-5.4) rectangle ++(.3,.15) ;
		\draw [rounded corners=0pt,fill=red!100](-19-5-1,-6.4) rectangle ++(.7,.3) ;

	\node [] at (-18-1.5,-6) {\LARGE $\hdots$};	
		
		\draw[thick ,fill=green!10] (-9-8,-6.5) rectangle ++(2.7,2) node[midway]{\footnotesize $\mathrm{U}_{K,1}$};
		\draw[thick] (-8-8+.35,-4.5)--+(0,.7);
		\draw[thick,fill=gray!15] (-8-8+.35,-3.8)--++(.75,0.5)--++(-1.5,0)--++(.75,-.5)--++(0,-.1);
		
		\draw [rounded corners=0pt,fill=white!10](-10-8,-6.4) rectangle ++(.7,1) ;
		\draw [rounded corners=0pt,fill=white!10](-9.8-8,-5.4) rectangle ++(.3,.15) ;
		\draw [rounded corners=0pt,fill=green!100](-10-8,-6.4) rectangle ++(.7,.75) ;

		\draw[thick ,fill=green!10] (.5-7+10,-6.5) rectangle ++(2.7,2) node[midway]{\footnotesize $\mathrm{U}_{K,M}$};
		\draw[thick] (1.5-7+.35+10,-4.5)--+(0,.7);
		\draw[thick,fill=gray!15] (1.5-7+.35+10,-3.8)--++(.75,0.5)--++(-1.5,0)--++(.75,-.5)--++(0,-.1);
		
		\draw [rounded corners=0pt,fill=white!10](-.5-7+10,-6.4) rectangle ++(.7,1) ;
		\draw [rounded corners=0pt,fill=white!10](-.3-7+10,-5.4) rectangle ++(.3,.15) ;
		\draw [rounded corners=0pt,fill=yellow!90](-.5-7+10,-6.4) rectangle ++(.7,.55) ;

	\node [] at (.5-7+10-2.5,-6) {\LARGE $\hdots$};

		\draw[thick ,fill=green!10] (.5-7+10-7,-6.5) rectangle ++(2.7,2) node[midway]{\footnotesize $\mathrm{U}_{1,M}$};
	\draw[thick] (1.5-7+.35+10-7,-4.5)--+(0,.7);
	\draw[thick,fill=gray!15] (1.5-7+.35+10-7,-3.8)--++(.75,0.5)--++(-1.5,0)--++(.75,-.5)--++(0,-.1);
	
	\draw [rounded corners=0pt,fill=white!10](-.5-7+10-7,-6.4) rectangle ++(.7,1) ;
	\draw [rounded corners=0pt,fill=white!10](-.3-7+10-7,-5.4) rectangle ++(.3,.15) ;
	\draw [rounded corners=0pt,fill=green!90](-.5-7+10-7,-6.4) rectangle ++(.7,.65) ;

		\draw[->,line width=1,green!100] (-18,8.3)--+(-4,-4);
		\draw[<-,dashed,line width=1,blue!100] (-18.5,8.7)--+(-4,-4);
		
		\draw[<-,line width=1,black!100] (-17.4,8.5)--+(-4.6,-10.7);
		\draw[->,line width=1,green!100] (-16.8,8.3)--+(-4.6,-10.7);

		\draw[<-,line width=1,dashed,blue!100] (-16,8.3)--+(-.3,-4);
		\draw[->,line width=1,green!100] (-15.4,8.1)--+(-.3,-4);
	%	
	
%    	\draw[<-,line width=1,dashed,blue!100] (-16+1,8.3+.5)--+(10.5,-4);
%    	\draw[->,line width=1,green!100] (-15.4+1,8.7+.5)--+(10.5,-4);

		\draw[<-,dashed,line width=1,blue!100] (-8.6+10,8.4)--+(-3.4,-3);
		\draw[->,line width=1,green!100] (-8+10,8.2)--+(-3.4,-3);
		
		\draw[<-,dashed,line width=1,blue!100] (-6.4+10,8.1)--+(2.5,-4.5);
		\draw[->,line width=1,green!100] (-5.8+10,8.5)--+(2.5,-4.5);
		
		\draw[<-,line width=1,black!100] (-7.5+10,8.5)--+(1.2,-9.7);
		\draw[->,line width=1,green!100] (-6.9+10,8.3)--+(1.2,-9.7);

%		\draw[<-,dashed,line width=1,blue!100] (-8.6+8,8.6)--+(-14,-4);
%    	\draw[->,line width=1,green!100] (-8+8,8.2)--+(-14,-4);	

		\draw[fill=black!100] (-7.8+10,3.5) circle (.03 cm);
	\draw[fill=black!100] (-7+10,3.25) circle (.03 cm);
	\draw[fill=black!100] (-6.2+10,3) circle (.03 cm);

		\draw[->,line width=1,green!100] (-14-11.7+4.3+5,-11.6-7+9-.5)--+(4.1,0);
		\draw[->,line width=1,dashed,blue!100] (-18-12+4+5.3,-12.85-7+9-.5)--+(8.35,0);
		\draw[->,line width=1,black!100] (-18-12+4+5,-14.05-7+9-.5)--+(8.7,0);

		\node [] at (-18-12.6+4+5,-11.5-7+9-.5) {\small Radar/WPT waveform};
		\node [] at (-21-11.85+4+5.21,-12.8-7+9-.5) {\small Echo signal};
		\node [] at (-10.6-13.15,-14.1-7+9-.5) {\small Data signal};

	\end{tikzpicture}
	\caption{A multi-UAV enabled integrated sensing and wireless powered communication. The ground users, i.e., $\mathrm{U}_{k,m}$, can be considered as traffic/air quality monitoring or fire detector sensors in smart cities, crop/soil monitoring or livestock tracking sensors in farmlands, etc. }
	\label{ht}
	\centering
\end{figure}
Then, the coordinates of $\mathrm{U}_{k,m}$ are modeled as 
\begin{eqnarray}
x_{k,m}=\bar{x}_{k,m} + \Delta x_{k,m},
~~~~
y_{k,m}=\bar{y}_{k,m} + \Delta y_{k,m},
\end{eqnarray}
respectively, where $\bar{x}_{k,m}$ and $\bar{y}_{k,m}$ are the user location estimates
available at UAVs, and $\Delta x_{k,m}$ as well as $\Delta y_{k,m}$ denote the associated location uncertainties. Furthermore, we assume that UAVs
know their own location perfectly \cite{xie2021uav,xie2018throughput,xu2018uav,cui2018robust}. Then, letting
$
\mathbf{r}_{k,m} = [x_{k,m} , y_{k,m}]^T,
\hspace{2pt}
\bar{\mathbf{r}}_{k,m} = [\bar{x}_{k,m} , \bar{y}_{k,m}]^T$, and $\Delta \mathbf{r}_{k,m} = [\Delta x_{k,m} , \Delta y_{k,m}]^T,
$ the circular uncertainty region can be represented as
$\Delta \mathbf{r}^T_{k,m}  \Delta \mathbf{r}_{k,m} \leq \bar{d}^2_{k,m}$, where $\bar{d}_{k,m}$ denotes the radius of the uncertainty region \cite{boshkovska2017robust}. The distance from $\mathrm{UAV}_m$ to $\mathrm{U}_{k,i}$ in time slot $n$ can be expressed as 
\begin{figure}
	\centering
	\begin{tikzpicture}[even odd rule,rounded corners=2pt,x=12pt,y=12pt,scale=.55,every node/.style={scale=.7}]
		\draw[thick ,fill=white!10,rounded corners=0pt] (\wi+0,\he+0) rectangle ++(6,2) node[midway]{ time slot 1};
		
		\draw[<->,line width=1] (\wi+0,-.6+\he)--+(6,0);
		
		\node [] at (\wi+3,-1.3+\he) {\small $\delta_t$};

		\draw[thick ,fill=white!10,rounded corners=0pt] (\wi+6,0+\he) rectangle ++(6,2) node[midway]{ \dots};

		\draw[thick ,fill=white!10,rounded corners=0pt] (\wi+12,0+\he) rectangle ++(6,2) node[midway]{ time slot n};
		
		\draw[<->,line width=1] (\wi+12,-.6+\he)--+(6,0);
		
		\node [] at (\wi+15,-1.3+\he) {\small $\delta_t$};

		\draw[thick ,fill=white!10,rounded corners=0pt] (\wi+18,0+\he) rectangle ++(6,2) node[midway]{ \dots};

		\draw[thick ,fill=white!10,rounded corners=0pt] (\wi+24,0+\he) rectangle ++(6,2) node[midway]{ time slot N};
		
		\draw[<->,line width=1] (\wi+24,-.6+\he)--+(6,0);
		
		\node [] at (\wi+27,-1.3+\he) {\small $\delta_t$};

		\draw[thick ,fill=green!10,rounded corners=0pt] (\wi+4,-5+\he) rectangle ++(7,2) node[midway]{sensing/WPT};
		
		\draw[<->,line width=1] (\wi+4,-5.6+\he)--+(7,0);
		
		\node [] at (\wi+7.5,-6.4+\he) {\small $\tau_0 \delta_t$};

		%        \draw[<->,line width=1] (\wi+4,-5.6+\he)--+(7,0);
		%        
		%        \node [] at (\wi+7.5,-6.4+\he) {\small $\tau_0 \hspace{1pt}\delta_t$};
		
		\node [] at (\wi+7.5,-7.6+\he) {\small $\vdots$};
		
		\draw[thick ,fill=red!10,rounded corners=0pt] (\wi+11,-5+\he) rectangle ++(6,2) node[midway]{$\mathrm{U}_{1,1}$ WIT};
		
		\draw[<->,line width=1] (\wi+11,-5.6+\he)--+(6,0);
		
		\node [] at (\wi+14,-6.4+\he) {\small $\tau_{1,1}[n] \delta_t$};

		%        \draw[<->,line width=1] (\wi+11,-5.6+\he)--+(6,0);
		%        
		%        \node [] at (\wi+14,-6.4+\he) {\small $\tau_{1,1}[n] \delta_t$};
		
		\node [] at (\wi+14,-7.6+\he) {\small $\vdots$};
		
		\draw[thick ,fill=red!10,rounded corners=0pt] (\wi+17,-5+\he) rectangle ++(3,2) node[midway]{$\hdots$};

		\draw[thick ,fill=red!10,rounded corners=0pt] (\wi+20,-5+\he) rectangle ++(6,2) node[midway]{$\mathrm{U}_{K,1}$ WIT};
		
		\draw[<->,line width=1] (\wi+20,-5.6+\he)--+(6,0);
		
		\node [] at (\wi+23,-6.4+\he) {\small $\tau_{K,1}[n] \delta_t$};

		%        \draw[<->,line width=1] (\wi+20,-5.6+\he)--+(6,0);
		%        
		%        \node [] at (\wi+23,-6.4+\he) {\small $\tau_{K,1}[n] \delta_t$};
		
		\node [] at (\wi+23,-7.6+\he) {\small $\vdots$};
		
		\draw[-,line width=1] (\wi+12,0+\he)--+(-8,-3);
		\draw[-,line width=1] (\wi+18,0+\he)--+(8,-3);
		
		\draw[-,line width=1] (\wi+26,-3+\he)--+(0,-6);
		\draw[-,line width=1] (\wi+4,-3+\he)--+(0,-6);

		\draw[thick ,fill=green!10,rounded corners=0pt] (\wi+4,-5+\he-6) rectangle ++(7,2) node[midway]{sensing/WPT};
		
		\draw[<->,line width=1] (\wi+4,-5.6+\he-6)--+(7,0);
		
		\node [] at (\wi+7.5,-6.4+\he-6) {\small $\tau_0 \hspace{1pt}\delta_t$};
		
		\draw[thick ,fill=red!10,rounded corners=0pt] (\wi+11,-5+\he-6) rectangle ++(6,2) node[midway]{$\mathrm{U}_{1,M}$ WIT};
		
		\draw[<->,line width=1] (\wi+11,-5.6+\he-6)--+(6,0);
		
		\node [] at (\wi+14,-6.4+\he-6) {\small $\tau_{1,M}[n] \delta_t$};
		
		\draw[thick ,fill=red!10,rounded corners=0pt] (\wi+17,-5+\he-6) rectangle ++(3,2) node[midway]{$\hdots$};
		
		\draw[thick ,fill=red!10,rounded corners=0pt] (\wi+20,-5+\he-6) rectangle ++(6,2) node[midway]{$\mathrm{U}_{K,M}$ WIT};
		
		\draw[<->,line width=1] (\wi+20,-5.6+\he-6)--+(6,0);
		
		\node [] at (\wi+23,-6.4+\he-6) {\small $\tau_{K,M}[n] \delta_t$};

		\draw[-,thick] (\wi+0,0+\he)--+(0,-3);
		\draw[-,thick] (\wi+30,0+\he)--+(0,-3);
		
		\node [] at (\wi+0,-3.5+\he) {\small $t=0$};
		\node [] at (\wi+30,-3.5+\he) {\small $t=T$};
		
	\end{tikzpicture}
	\caption{Protocol structure for multi-UAV enabled integrated sensing and wireless powered communication.}
	\label{ht1}
	\centering
\end{figure}
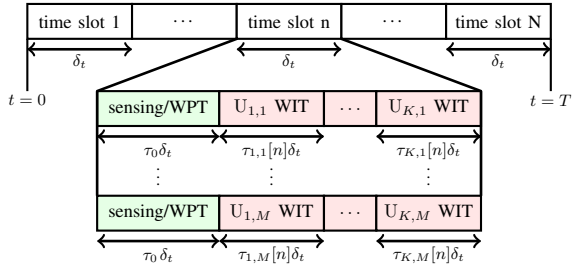
\begin{equation}
d_{k,i,m} [n]= \sqrt{ {\parallel \widetilde{\mathbf{q}}_m [n] - \mathbf{r}_{k,i} \parallel }^2_2 + z_m^2 [n]},
\end{equation}
where $\widetilde{\mathbf{q}}_m [n]= [x_m[n], y_m[n]]^T $. Theoretical investigations and empirical studies indicate that areas with ample UAV coverage and minimal ground obstructions, such as rural regions, exhibit a notably high probability of LoS links \cite{7108163}. Assuming the UAVs in our model operate at moderately high altitudes, transmission between UAVs and ground users encounters minimal obstruction or reflection. Therefore, we assume that the communication links from
UAVs to the ground users are dominated by the LoS links \cite{xie2020common,xie2018throughput,8588990}. The channel power gain from $\mathrm{UAV}_m$ to $\mathrm{U}_{k,i}$ during
slot $n$ follows the free-space path loss model which can be written as $h_{k,i,m} [n]=\rho_0 d^{-2}_{k,i,m} [n]$, where $\rho_0$ denotes the channel power at the reference distance $d_0 = 1$ m \cite{xie2020common,xie2018throughput,8588990}. 

Fig.~\ref{ht1} shows the timing protocol of the proposed method. Followed by harvest-then-transmit protocol, each time slot of width $\delta_t$ is divided into $K+1$ subslots where the first subslot with duration $\tau_0 \delta_t$ is used for sensing/WPT and the other subslots of duration $\tau_{k,m} [n] \delta_t$ are allocated for uplink transmission from users to UAVs. 
\subsection{Downlink: Sensing/WPT Phase}
\subsubsection{Sensing perspective}\label{mmnh} The sensing interval consists of $L$ successive transmission of PMCW radar sequences $\widetilde{\mathbf{x}}_m \in \complexC^{\widetilde{N}},~1\leq m \leq M,$ with total transmission power of $p_m^{dl}[n]$. More precisely, each sequence $\widetilde{\mathbf{x}}_m$ comprises $\widetilde{N}$ sinusoidal signals, each characterized by a distinct initial phase. Note that we consider $\widetilde{N}$ as a fixed parameter  determined as per radar range resolution and other practical considerations. The $l$th received signal of a target associated with a given $\widetilde{\mathbf{x}}_m$ in $\mathrm{UAV}_m$ at the cell under test can be modeled as \cite{6650043}
\begin{align}  \label{en1}
	\widetilde{\mathbf{y}}_{m,l} [n]=& \alpha_{m} [n] \widetilde{\mathbf{x}}_m  +\sum_{k=-\widetilde{N}+1 ,k\neq 0}^{\widetilde{N}-1} \widetilde{\alpha}_{m,k} [n] \mathbf{J}_k  \widetilde{\mathbf{x}}_m \\ \nonumber&+\mathbf{n}_{m,l} [n],~ \forall n,m, 1 \leq l \leq L,
\end{align}
where $\alpha_{m} [n]$ and $\widetilde{\alpha}_{m,k} [n]$ are the complex parameter corresponding to the propagation and backscattering effects from the desired and interfering targets, respectively, $\mathbf{n}_{m,l} [n]\sim \mathcal{CN} \left(\mathbf{0},{{\sigma}_{m,l}^2}{\mathbf{I}}_{\widetilde{N}} \right)$ is a noise vector, and $\mathbf{J}_k$ denotes the periodic shift matrix 
\begin{equation}
	\mathbf{J}_k=\begin{bmatrix}
		\mathbf{0}_{(\widetilde{N}-k)\times k} & \mathbf{I}_{\widetilde{N}-k} \\ \mathbf{I}_{k} & \mathbf{0}_{k \times (\widetilde{N}-k)}
	\end{bmatrix}, k\geq 1, ~ \mathbf{J}_{-k}=\mathbf{J}^T_k.
\end{equation}
Note that we do not consider the effect of Doppler shift for intra-pulse code $\widetilde{\mathbf{x}}_m$. That is, this effect is incorporated in complex coefficients  $\alpha_{m} [n]$ and $\widetilde{\alpha}_{m,k} [n]$ as a constant phase. However, the inter-pulse Doppler shift is non-negligible and will be taken into account in the following. Besides, we ignore the interference from other UAV radars, i.e., $\mathrm{UAV}_i,~ \forall i \neq m$, due to their weak echo signals. This assertion is supported by two key factors: i) each radar antenna's FoV is tailored to cover only its respective cluster, limiting interference signals from other clusters to enter the radar receiver through its sidelobes, resulting in significant attenuation in the range of approximately 30 dB, as demonstrated in \cite{ 7858652 }; ii) the potential to employ pseudo-orthogonal code design and leverage different frequencies for radar transmit waveforms for different UAVs facilitates an approximately 20 dB attenuation of interferences, as illustrated in \cite{ RAEI2023108914 }. Then, the received signal $\widetilde{\mathbf{y}}_{m,l} [n]$ is processed via $\widetilde{\mathbf{w}}_m \in \complexC^{\widetilde{N}}$, viz. range processing. The time delay at which each receive filter output signal (i.e., the correlation between $\widetilde{\mathbf{w}}_m$ and $\widetilde{\mathbf{y}}_m$) has its maximum value can be used to estimate the target range. As shown in Fig.~\ref{hh1}, following the range processing, Doppler processing is applied on each range-cell to obtain target speed via frequency analysis (usually implemented by FFT) of $L$ samples associated with the range cell. Therefore, in practice, it is better to choose $L$ as $2^{i}, \hspace{1pt}lb \leq i \leq ub$, where $lb$ and $ub$ are respectively determined according to the minimum required sensing SINR and the allowed system complexity. As illustrated in Fig.~\ref{ttt}, the sensing/WPT subslot duration can be obtained as $\tau_0 \delta_t=L \widetilde{\tau}_0\delta_t$ where $\widetilde{\tau}_0$ is a fix parameter that can be determined by $\widetilde{N}$ and the sampling frequency of radars. Using \eqref{en1}, the sensing SINR of  $\mathrm{UAV}_m$ after range processing block, i.e., the signal $\bar{y}_{m,l}[n]={\widetilde{\mathbf{w}}}^H_m \widetilde{\mathbf{y}}_{m,l}[n]$ can be written as
\begin{figure}
	\centering
	\begin{tikzpicture}[even odd rule,rounded corners=2pt,x=12pt,y=12pt,scale=.55,every node/.style={scale=.7}]
		
		\draw[->,line width=1] (-1.5,0)--+(4,0);
		
		\node [] at (2-1.5,1) { $\widetilde{\mathbf{y}}_{m,L}[n]$};

		\draw[thick ,fill=white!10,rounded corners=0pt] (4-1.5,-1) rectangle ++(5.5,2) node[midway]{ $\widetilde{\mathbf{w}}_{m} \in \complexC^{\widetilde{N}}$};

		\draw[->,line width=1] (8,0)--+(10,0);
		
		\node [] at (13,1) { $\bar{{y}}_{m,L}[n],~\mathrm{SINR}_{m,L}[n]$};

		\node [] at (6-1.5,4) {\huge $\vdots$};

		\draw[->,line width=1] (0-1.5,0+8)--+(4,0);
		
		\node [] at (2-1.5,1+8) { $\widetilde{\mathbf{y}}_{m,1}[n]$};

		\draw[thick ,fill=white!10,rounded corners=0pt] (4-1.5,-1+8) rectangle ++(5.5,2) node[midway]{ $\widetilde{\mathbf{w}}_{m} \in \complexC^{\widetilde{N}}$};

		\draw[->,line width=1] (8,0+8)--+(10,0);
		
		\node [] at (13,1+8) { $\bar{{y}}_{m,1}[n],~\mathrm{SINR}_{m,1}[n]$};

		\draw[thick ,fill=white!10,rounded corners=0pt] (18,-1) rectangle ++(3.25,10) node[midway]{ Buffer};
		
		\draw[->,line width=1] (21.25,4)--+(2,0);
		
		%		\draw[->,line width=1] (18,0+8)--+(2.9,-3.75);
		%		
		%		
		%		\draw[->,line width=1] (18,0)--+(2.9,+3.75);

		\draw[thick ,fill=white!10,rounded  corners=0pt] (21+2.25,2.75) rectangle ++(4.5,2.5) node[midway]{ $\widetilde{\mathbf{a}}_m \in \complexC^{L}$};
		
		\draw[->,line width=1] (25+2.75,0+4)--+(5.5,0);
		
		\node [] at (27.75+2.75,5) { $\widetilde{\mathrm{SINR}}_{m}[n]$};

		\node [] at (23+2.25,-2.5) {Doppler processing unit};
		
		\node [] at (23-17-1.5,-2.5) {range processing unit};

	\end{tikzpicture}
	\caption{The block diagram for the radar receiver of the $\mathrm{UAV}_m$ considering integration over $L$ received signal.}
	\label{hh1}
	\centering
\end{figure}
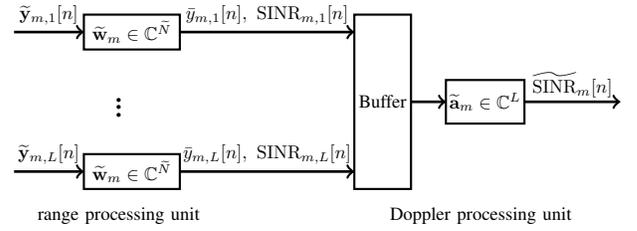
\begin{figure} 
	\centering
	\includegraphics[scale=.265]{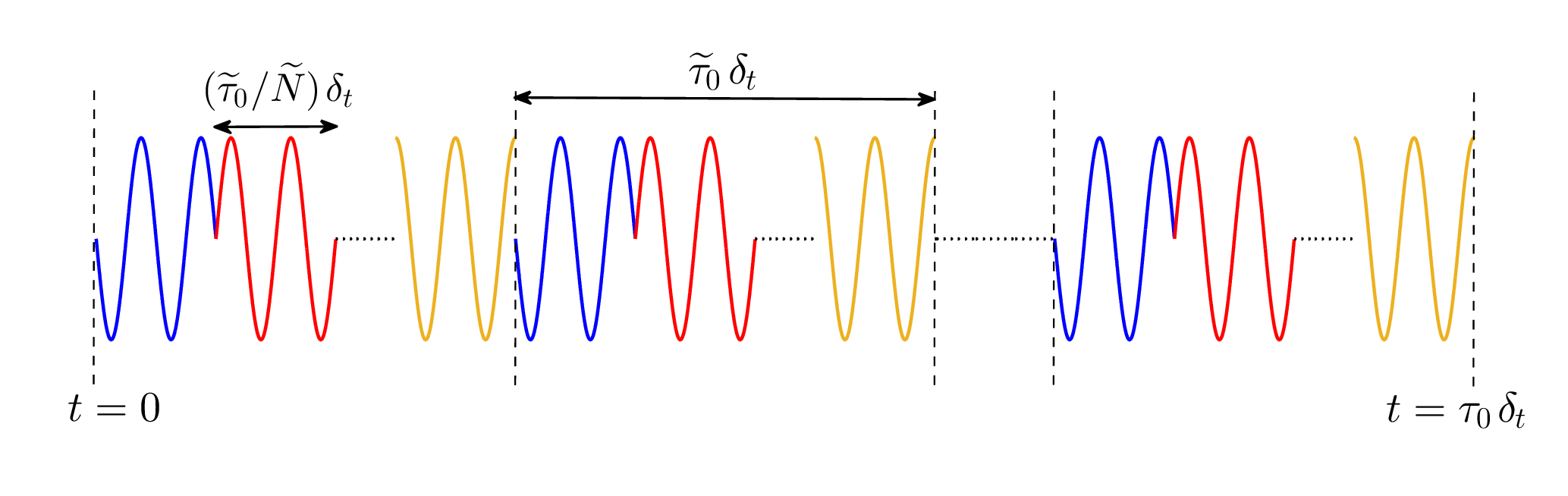}
	\caption{ Illustration of sensing/WPT waveform for the pulse-modulated PMCW radar. The total sensing/WPT duration consists of $L$ successive transmission of radar sequence each with duration $\widetilde{\tau}_0 \delta_t$. Each radar sequence itself consists of $\widetilde{N}$ sub-sequences with duration $(\widetilde{\tau}_0/\widetilde{N})\delta_t$ which are shown by different colors. Specifically, the signals shown in different colors are sinusoidal signals, each with a distinct initial phase. }
	\label{ttt}
	\centering
\end{figure}
\begin{equation} \label{kkl}
	{{\mathrm{SINR}}}_{m,l}\hspace{1pt}[n]= \frac{  {\vert \alpha_{m}[n] \vert}^2  {\vert \widetilde{\mathbf{w}}_m^H \widetilde{\mathbf{x}}_m \vert}^2} {\widetilde{\mathbf{w}}_m^H \hm{\Xi}_{m,l} \widetilde{\mathbf{w}}_m},
\end{equation}
where
\begin{align} \label{kkl21}
	\hm{\Xi}_{m,l}=& {{\sigma}_{m,l}^2}{\mathbf{I}}_{\widetilde{N}} + \sum_{k=-\widetilde{N}+1 ,k\neq 0}^{\widetilde{N}-1} \widetilde{\sigma}_{m,k}^2 \mathbf{J}_k \widetilde{\mathbf{x}}_m \widetilde{\mathbf{x}}_m^H  \mathbf{J}^H_k,
\end{align}
is a positive definite matrix and $\widetilde{\sigma}_{m,k}^2=\mathbb{E} [\vert \widetilde{\alpha}_{m,k} [n] \vert^2]$. Then, assuming a-priori known target Doppler shift\footnote{Several techniques assume a-priori known Doppler frequency (see e.g., \cite{de2011design,naghsh2013unified}); however, in practice, the Doppler shift can be estimated at the receiver, e.g., via a bank of filters matched to different Doppler frequencies \cite{stoica2008transmit}.}, the total sensing SINR of $\mathrm{UAV}_m$ after Doppler processing block can be obtained as
\begin{equation} \label{kkl0}
	\widetilde{{\mathrm{SINR}}}_{m}\hspace{1pt}[n]= \frac{ \beta_m {\vert \alpha_{m}[n] \vert}^2  {\vert \widetilde{\mathbf{w}}_m^H \widetilde{\mathbf{x}}_m \vert}^2} {\widetilde{\mathbf{w}}_m^H \widetilde{\hm{\Xi}}_{m} \widetilde{\mathbf{w}}_m},
\end{equation}
where
\begin{equation}
	\beta_m=L+ \sum_{i=1}^{L}\sum_{j=1, j \neq i}^{L}
	[\widetilde{\mathbf{a}}_m^{*}]_i \hspace{1pt} [\widetilde{\mathbf{a}}_m]_j,
\end{equation}
with $\widetilde{\mathbf{a}}_m \in \complexC^{L}$ is the Doppler processing filter, and 
\begin{align} \label{kkl210}
	\widetilde{\hm{\Xi}}_{m}=& \sum_{l=1}^{L} {{\sigma}_{m,l}^2}{\mathbf{I}}_{\widetilde{N}} +\beta_m \sum_{k=-\widetilde{N}+1 ,k\neq 0}^{\widetilde{N}-1} \widetilde{\sigma}_{m,k}^2 \mathbf{J}_k \widetilde{\mathbf{x}}_m \widetilde{\mathbf{x}}_m^H  \mathbf{J}^H_k,
\end{align}
is a positive definite matrix. For instance, the filter $\widetilde{\mathbf{a}}_m$ can be obtained as $\widetilde{\mathbf{a}}_m=[1, e^{-j\nu_m},...,$ $e^{-j{\nu_m}(L-1)}]^T$ as a simple FFT-based Doppler processing filter with $\nu_m$ being the normalized Doppler frequency between the $\mathrm{UAV}_m$ and a possible target.

As mentioned, parameter  $\alpha_{m}[n]$ encapsulates information regarding target propagation and backscattering effects. Swerling models offer statistical characterizations of target behavior, ranging from Swerling 0 to Swerling \rom{4}. Swerling 0 denotes a non-fluctuating target model, while Swerling \rom{1} through Swerling \rom{4} represent various fluctuating models. Swerling \rom{1} applies to scenarios where the target's velocity is relatively low compared to the observation time, effectively rendering it stationary. Given our focus on low-speed malicious UAVs as target entities, Swerling \rom{1} serves as an appropriate assumption within our model. Therefore, $\widetilde{{\mathrm{SINR}}}_{m}\hspace{1pt}[n]$ is assumed to be constant over time slots, i.e., $\widetilde{{\mathrm{SINR}}}_{m}\hspace{1pt}[n]=\widetilde{{\mathrm{SINR}}}_{m}\hspace{1pt}$.
\subsubsection{WPT perspective}
The harvested energy at $\mathrm{U}_{k,m}$ for linear EH model can be expressed as\footnote{In practice, the minimum required energy for typical EH circuits can be in the range of $1\sim10~\mu$W \cite{clerckx2018fundamentals}.}
\begin{equation} \label{eee}
E_{k,m} [n]=\tau_0 \delta_t \epsilon_{k,m} \sum_{i=1}^{M} h_{k,m,i} [n] p_i^{dl}[n],
\end{equation}
where $\epsilon_{k,m}$ denotes the energy conversion efficiency. A linear EH circuit is considered here to make the model more tractable. We explore the concept of non-linear EH in Subsection~\ref{nonl}.
\subsection{Uplink: WIT Phase}
In the uplink phase, the transmit power of $\mathrm{U}_{k,m}$ in time slot $n$ is denoted by $p^{{ul}}_{k,m} [n]$. The following EH constraint for $\mathrm{U}_{k,m}$ at time slot $n$ should be satisfied: 
\begin{align} \label{key}
	& \scalemath{.97}{\tau_{k,m}[n] \delta_t  p^{{ul}}_{k,m} [n]  \leq \sum^{n}_{j=1} E_{k,m} [j] + E^0_{k,m} - \sum^{n-1}_{j=1} \tau_{k,m}[j] \delta_t  p^{{ul}}_{k,m} [j],}
\end{align}
where $ E^0_{k,m}$ is the remaining
energy for $\mathrm{U}_{k,m}$ from previous time periods. It is assumed that $E^0_{k,m}$ is associated with sensing-only mode where the system works before employing the proposed protocol in $[0 , T]$. Moreover, some parts of $E^0_{k,m}$ can be obtained by solar energy where $\mathrm{U}_{k,m}$ has a hybrid solar-RF energy collector circuit \cite{tran2022hybrid}. Note that $E^0_{k,m}$ guarantees the reliable energy for uninterrupted communication in the WIT phase. Then, the achievable throughput of $\mathrm{U}_{k,m}$ in time slot $n$ is given by \cite{xie2020common}
\begin{equation}  \label{sinr}
	R_{k,m} [n]=\tau_{k,m} [n] \delta_t \textrm{log}_2 \left(1+ \frac{p^{{ul}}_{k,m} [n] h_{k,m,m}[n]}{ {\sigma_{c,m}^2}} \right),
\end{equation}
where ${\sigma_{c,m}^2}$ is the power of the additive white Gaussian noise at the  communication receiver of $\mathrm{UAV}_{m}$. Notice that UAVs employ CoMP reception techniques \cite{xie2020common,liu2019comp} and as a result, the inter-cluster interference is removed. Thus, the average throughput of $\mathrm{U}_{k,m}$ over $N$ time slots is given by $R_{k,m}=\frac{1}{N} \sum^{N}_{n=1} R_{k,m} [n]$.	
\section{Problem Formulation And The Proposed Method}\label{sum}
	In this section, we cast the optimization problem in which we aim to jointly maximize the minimum radar SINR, i.e., $\min_m \widetilde{\mathrm{SINR}}_m$, and minimum throughput of communication users, i.e., $\min_{k,i} R_{k,i}$, in our model by designing the radar/WPT waveforms $\widetilde{\mathbf{X}}=\lbrace \widetilde{\mathbf{x}}_m ,\hspace{2pt} \forall m \rbrace$, radar receive filters $\widetilde{\mathbf{W}}= \lbrace \widetilde{\mathbf{w}}_m ,\hspace{2pt} \forall m \rbrace$, time scheduling parameters $\mathbf{A}= \lbrace \tau_{k,m} [n],\hspace{2pt} \forall k,m,n \rbrace$, integration parameter $L$, uplink power of users $\mathbf{P}^{ul}= \lbrace  p^{ul}_{k,m} [n],\hspace{2pt} \forall k,m,n \rbrace$, and UAV trajectories $\mathbf{Q}= \lbrace \mathbf{q}_m [n],\hspace{2pt} \forall m,n \rbrace$. To find the Pareto-optimal solutions of the mentioned multi-objective problem, we adopt the scalarization technique \cite{boyd2004convex} using the Pareto weight $\mu \in [0 ,1]$ as follows
\begin{align}\label{maxmin1}
	&\hspace{-10pt} \max_{\widetilde{\mathbf{X}}, \widetilde{\mathbf{W}},{\mathbf{A}}, L, \mathbf{P}^{ul}, \mathbf{Q} }  ~ (1-\mu) \displaystyle \min_{\substack{{1 \leq k \leq  {K}}\\ {1 \leq i \leq  {M}}}} ~ \min_{  \Delta \mathbf{r}^T_{k,i}  \Delta \mathbf{r}_{k,i} \leq \bar{d}^2_{k,i} } R_{k,i} \\ \nonumber  & \hspace{55pt}+\mu \displaystyle \min_{1 \leq m \leq  {M}  } \hspace{1pt} \widetilde{\mathrm{SINR}}_m
	\\ \nonumber 
	&\hspace{-7pt} \mbox{s.t.}\\ \nonumber  & \scalemath{.97}{\textrm{C}_{1}:0 \leq \tau_{k,m} [n] \leq 1-L\widetilde{\tau}_0, \hspace{3pt}  \sum_{k=1}^{K} \tau_{k,m} [n] \leq 1- L\widetilde{\tau}_0, \hspace{2pt} \forall k,m,n,} \\ \nonumber & 
	\scalemath{.97}{\textrm{C}_{2}: {\parallel \mathbf{q}_m [n] - \mathbf{q}_m [n-1] \parallel}_2 \leq \delta_t  {v}_{\textrm{max}}, \hspace{2pt} z_{\textrm{min}} \leq z_m [n] \leq z_{\textrm{max}} , } \\ \nonumber & \hspace{14pt} \scalemath{.95}{\frac{1}{N} \sum_{n=1}^{N} z_m [n] \geq z^{\mathrm{tr}}_m, \hspace{2pt} \widetilde{\mathbf{q}}_m [n] \in \mathrm{CR}_m
	, \hspace{2pt} \mathbf{q}_m [0] = \mathbf{q}_m[N] ,\hspace{2pt} \forall m,n,}
	\\ \nonumber & 
	\textrm{C}_{3}: {\parallel \widetilde{\mathbf{q}}_m [n] - \mathbf{q}^{\mathrm{NFZ}}_{m,j} \parallel}^2_2 \geq \left(r^{\mathrm{NFZ}}_{m,j}\right)^2, \hspace{2pt} \forall m,n, \hspace{2pt} 1 \leq j \leq N_m^{\mathrm{NFZ}} , \\ \nonumber & \hspace{19pt}
	{\parallel \widetilde{\mathbf{q}}_m [n] - \widetilde{\mathbf{q}}_{m^{\prime}} [n] \parallel}^2_2 \geq d_{\textrm{min}}^2, \hspace{2pt} \forall m \neq m^{\prime},n,
	\\ \nonumber & 
	\scalemath{.97}{\textrm{C}_{4}:L \in \lbrace 2^{lb}, 2^{lb+1}, ..., 2^{ub} \rbrace,~ \textrm{C}_{5}: \vert \widetilde{{x}}_m(i) \vert^2 = p_m^{dl} [n], \hspace{2pt} \forall m,i,n,}  \\ \nonumber &   
	\textrm{C}_{6}: \sum^{n}_{j=1}\tau_{k,m}[j] \delta_t  p^{{ul}}_{k,m} [j]  \leq   \min_{  \Delta \mathbf{r}^T_{k,m}  \Delta \mathbf{r}_{k,m} \leq \bar{d}^2_{k,m} } \sum^{n}_{j=1} E_{k,m} [j]  \\ \nonumber & \hspace{15pt} + E^0_{k,m}, \hspace{2pt} \forall k,m,n,
	%		\\ \nonumber \;\;&& \hspace{-30pt}
	%		\textrm{C}_{5}: R_k \geq R_{QoS},\hspace{2pt} \forall k,	
\end{align}
where ${v}_{\textrm{max}}$ is the maximum speed of the UAVs and $\mathrm{CR}_m$ indicates the $m$th cluster region. Note that the design variables $\widetilde{\mathbf{X}}$, $L$, and  $\mathbf{Q}$ are the joint parameters between sensing and communication tasks. The constraint $\frac{1}{N}\sum_{n=1}^{N} z_m [n] \geq z^{\mathrm{tr}}_m$ in $\textrm{C}_2$ is considered to bring the desired coverage for a given radar FoV during the period of $T$ seconds. Precisely, $z^{\mathrm{tr}}_m$ can be determined numerically in such a way that the possible targets in $\mathrm{CR}_m$ are seen by the $\mathrm{UAV}_m$. The first part of the constraint $\textrm{C}_{3}$ introduces $N_m^{\mathrm{NFZ}}$ cylindrical NFZ areas in the $m$th cluster\footnote{NFZs are considered due to security, privacy, or safety reasons \cite{valavanis2015handbook,li2018joint}. Note that the NFZs can also be modeled as a polygonal \cite{lee2020uav}.} with coordinate center $\mathbf{q}^{\mathrm{NFZ}}_{m,j}$ and radius $r^{\mathrm{NFZ}}_{m,j}$. The second part of constraint $\textrm{C}_{3}$ is incorporated to ensure collision avoidance among UAVs, where $d_{\textrm{min}}$ represents the minimum allowable separation distance between them. Although constraint $\textrm{C}_{2}$ restricts each UAV's trajectory to remain within its designated cluster region, denoted by $\mathrm{CR}_m$, the additional condition in $\textrm{C}_{3}$ is essential for preventing collisions, particularly near the boundaries of adjacent clusters. Also, $\textrm{C}_{5}$ represents the unimodularity constraint of the transmit sequence.

It can be seen that the objective function and constraints $\textrm{C}_3-\textrm{C}_6$ are non-convex and so the problem. Therefore, in the following, we devise a method based on alternating optimization to deal with the non-convex design problem. Tackling subproblems corresponding to aforementioned alternating optimization is associated with novel tricks; e.g., by employing fractional programming, S-procedure, and MM to be discussed shortly.
\subsection{Maximization over $\widetilde{\mathbf{W}}$ for fixed $[\widetilde{\mathbf{X}}, \mathbf{A}, L, \mathbf{P}^{ul},\mathbf{Q}]$}\label{w}
We first consider the problem in \eqref{maxmin1} w.r.t. $\widetilde{\mathbf{W}}$ which is a unconstrained problem. Using the Cauchy-Schwartz inequality, we can write the following expression for $\widetilde{\mathrm{SINR}}_m$ in the objective function 
\begin{equation*} \scalemath{.93}{
	{\vert \widetilde{\mathbf{w}}_m^H \widetilde{\mathbf{x}}_m \vert}^2= {\vert \widetilde{\mathbf{w}}_m^H \widetilde{{\hm{\Xi}}}_m^{1/2} \widetilde{{\hm{\Xi}}}_m^{-1/2} \widetilde{\mathbf{x}}_m \vert}^2 \leq \left( \widetilde{\mathbf{w}}_m^H \widetilde{{\hm{\Xi}}}_m \widetilde{\mathbf{w}}_m \right) \left( \widetilde{\mathbf{x}}_m^H \widetilde{{\hm{\Xi}}}_m^{-1} \widetilde{\mathbf{x}}_m \right),}
\end{equation*}
where the equality holds for
\begin{equation} \label{jjk}
\widetilde{\mathbf{w}}_m={\widetilde{\hm{\Xi}}}_m^{-1} \widetilde{\mathbf{x}}_m,
\end{equation}
(by neglecting a multiplicative constant) which yields a closed-form solution for $\widetilde{\mathbf{w}}_m$. 
%Then, $\widetilde{\mathrm{SINR}}_m$ can be written as
%\begin{equation}
%\beta_m\vert \alpha_m \vert^2  \widetilde{\mathbf{x}}_m^H {\hm{\Xi}}_m^{-1} \widetilde{\mathbf{x}}_m.
%\end{equation}
%Thus, by introducing an auxiliary variable $\theta_a$, the problem in \eqref{maxmin1} for fixed $[{\mathbf{X}}, L, \mathbf{P}^{ul}, \mathbf{Q}]$ can be written in its epigraphic form as follows
%\begin{align}\label{maxmin2}
%	&~~\max_{{\mathbf{A}}, \theta_a } ~ \theta_a\\ \nonumber 
%	&\mbox{s.t.}~ \textrm{C}_{1},~\textrm{C}_{5}, \\ \nonumber 
%	&\scalemath{.93}{ \textrm{C}_{6}: \mu L \vert \alpha_T \vert^2  \widetilde{\mathbf{x}}_m^H {\hm{\Xi}}_m^{-1} \widetilde{\mathbf{x}}_m + (1-\mu) \hspace{-8pt}\min_{  \Delta \mathbf{w}^T_{k,i}  \Delta \mathbf{w}_{k,i} \leq \bar{d}^2_{k,i} } \hspace{-8pt} R_{k,i} \geq \theta_a,  \forall m,k,i.}	
%\end{align}
\subsection{Maximization over $[\widetilde{\mathbf{X}}, \mathbf{A}, \mathbf{P}^{ul},\mathbf{Q}]$ for fixed $[\widetilde{\mathbf{W}}, L]$}\label{joint}
Then, we consider the problem in \eqref{maxmin1} w.r.t. $[\widetilde{\mathbf{X}}, \mathbf{A}, \mathbf{P}^{ul},\mathbf{Q}]$ leading to the associated subproblem. The objective function and the constraints $\textrm{C}_3$, $\textrm{C}_5$, and $\textrm{C}_6$ are non-convex. Fig.~\ref{hh11} summarizes the procedure for dealing with this subproblem. First, we aim to deal with the sensing SINR in the objective function which is a non-convex fractional term. We rewrite the second term in the objective function as $f_1(\widetilde{\mathbf{x}}_m)= g_1 (\widetilde{\mathbf{x}}_m)/g_2 (\widetilde{\mathbf{x}}_m)$ for all $m$ where 
$
g_1(\widetilde{\mathbf{x}}_m)= \widetilde{\mathbf{x}}_m^H \hm{{\Gamma}}_m \widetilde{\mathbf{x}}_m
$
and 
$
g_2(\widetilde{\mathbf{x}}_m)=  \widetilde{\mathbf{x}}_m^H \widetilde{\hm{{\Gamma}}}_m \widetilde{\mathbf{x}}_m  +\gamma_m
$
with 
\begin{equation}
	\hm{\Gamma}_m=\mu \beta_m \vert \alpha_m \vert^2 \widetilde{\mathbf{w}}_m \widetilde{\mathbf{w}}_m^H,
\end{equation}
%\begin{equation}
%\widetilde{\hm{\Gamma}}_m=\vert \alpha_T \vert^2 \widetilde{\mathbf{w}}_m \widetilde{\mathbf{w}}_m^H,
%\end{equation}
\begin{equation}
	\widetilde{\hm{\Gamma}}_m= \beta_m \sum_{k=-\widetilde{N}+1 ,k\neq 0}^{\widetilde{N}-1} \widetilde{{\sigma}}_{m,k}^2 \mathbf{J}^H_k \widetilde{\mathbf{w}}_m \widetilde{\mathbf{w}}_m^H \mathbf{J}_k,
\end{equation}
and 
$
\gamma_m=\sum_{l=1}^{L} {\sigma_{m,l}^2} \hspace{2pt} \widetilde{\mathbf{w}}_m^H \widetilde{\mathbf{w}}_m.
$

\begin{prop} \label{cor1}
Let the objective function of the subproblem in \ref{joint} be as follows
\begin{equation}\label{mmb}
	q(\mathbf{\widetilde{X}})= a + \displaystyle \mu \min_{1\leq m \leq M} f_1(\mathbf{\widetilde{x}}_m),
\end{equation}
where
\begin{equation}
	a= (1-\mu) \displaystyle \min_{\substack{{1 \leq k \leq  {K}}\\ {1 \leq i \leq  {M}}}} ~ \min_{  \Delta \mathbf{r}^T_{k,i}  \Delta \mathbf{r}_{k,i} \leq \bar{d}^2_{k,i} } R_{k,i},
\end{equation}
is a constant term w.r.t. $\mathbf{\widetilde{X}}$.
By exploiting the idea of fractional programming \cite{dinkelbach1967nonlinear} and assuming $g_2(\widetilde{\mathbf{x}}_m) >0$ (to ensure that $f_1(\widetilde{\mathbf{x}}_m)$ has a finite value), it is proved that the objective function in \eqref{mmb} can be dealt with via iterative maximization of the below function w.r.t. $\mathbf{\widetilde{X}}$:
\begin{equation}
	q(\mathbf{\widetilde{X}})= a + \displaystyle \mu \min_{1\leq m \leq M} f_2(\mathbf{\widetilde{x}}_m),
\end{equation}
where
\begin{equation}
 f_2(\widetilde{\mathbf{x}}_m)= g_1(\widetilde{\mathbf{x}}_m) - f_1(\widetilde{\mathbf{x}}_m^{(\kappa-1)}) g_2(\widetilde{\mathbf{x}}_m). 
\end{equation}
	\end{prop}
	\begin{proof}
		Please refer to Appendix~\ref{app11}.
	\end{proof}
	Then, considering ${\parallel \widetilde{\mathbf{x}}_m \parallel }^2_2 =\widetilde{N} p_m^{dl}[n]$ from the unimodularity constraint $\textrm{C}_5$ in \eqref{maxmin1}, one can write
\begin{equation} \label{nnm}
	f_2(\widetilde{\mathbf{x}}_m)=\widetilde{\mathbf{x}}_m^H \hm{\Upsilon}_{m}^{(\kappa-1)}\widetilde{\mathbf{x}}_m,
\end{equation}
where
\begin{equation*}
	\hm{\Upsilon}_{m}^{(\kappa-1)}=
	\hm{{\Gamma}}_m- f_1(\widetilde{\mathbf{x}}_m^{(\kappa-1)}) \left (\widetilde{\hm{\Gamma}}_m+\frac{\gamma_m}{ \widetilde{N}p_m^{dl}[n]}\mathbf{I}_{\widetilde{N}} \right).
\end{equation*}
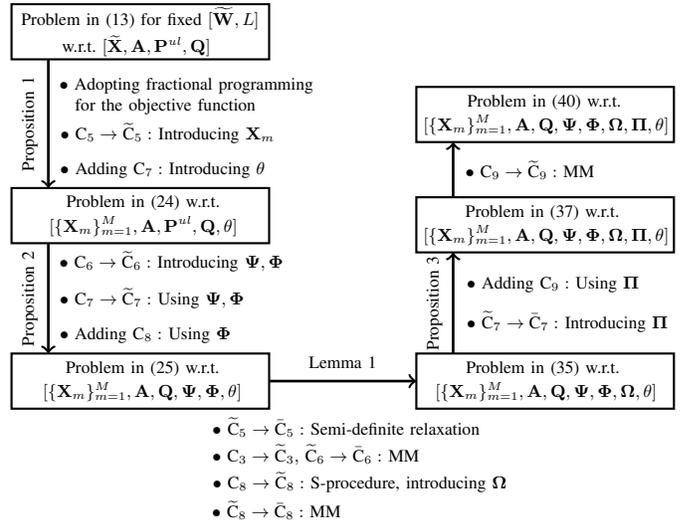
\begin{figure}
	\centering
	\begin{tikzpicture}[even odd rule,rounded corners=2pt,x=12pt,y=12pt,scale=.58,every node/.style={scale=.73}]

		\node[draw=white!20, rotate=90, minimum height=-1 cm, minimum width=1 cm] at (16, -13.25) {\small Proposition~\ref{cor3}};

		\node[draw=white!20, rotate=90, minimum height=-1 cm, minimum width=1 cm] at (-6, -13) {\small Proposition~\ref{cor2}};
		
		\draw[thick ,fill=white!10,rounded corners=0pt] (-7,0) rectangle ++(14,3) node[midway]{ \shortstack[l]{\small Problem in \eqref{maxmin1} for fixed $[\widetilde{\mathbf{W}}, L]$\\    \small \hspace{18pt} w.r.t. $[\widetilde{\mathbf{X}}, \mathbf{A}, \mathbf{P}^{ul},\mathbf{Q}]$}};
		
		\draw[->,line width=1] (-5,0)--+(0,-7);
		
		\node[draw=white!20, rotate=90, minimum height=-1 cm, minimum width=1 cm] at (-6, -3.5) {\small Proposition~\ref{cor1}};
		
		\node [] at (2.5,-2) {\shortstack[l]{ \small $\bullet$ Adopting fractional programming \\ \small \hspace{5pt} for the objective function}};
		
		\node [] at (2.5,-4) {\hspace{-23pt} \small$\bullet$ ${\textrm{C}}_5 \rightarrow \widetilde{\textrm{C}}_5:$ Introducing $\mathbf{X}_m$};
		
		\node [] at (2.5,-6) {\hspace{-28pt} \small$\bullet$ Adding ${\textrm{C}}_7:$ Introducing $\theta$};

		\draw[thick ,fill=white!10,rounded corners=0pt] (-7,-10) rectangle ++(14,3) node[midway]{ \shortstack[l]{\small \hspace{12pt}Problem in \eqref{maxmin3} w.r.t. \\    \small \hspace{0pt} $[\lbrace {\mathbf{X}}_m\rbrace_{m=1}^{M},\mathbf{A}, \mathbf{P}^{ul}, \mathbf{Q}, \theta]$}};
		
		\draw[->,line width=1] (-5,-10)--+(0,-6);

		\node [] at (2.5,-11) {\hspace{-19pt} \small$\bullet$ ${\textrm{C}}_6 \rightarrow \widetilde{\textrm{C}}_6:$ Introducing $\hm{\Psi}, \hm{\Phi}$};
		
		\node [] at (2.5,-13) {\hspace{-39pt} \small$\bullet$ ${\textrm{C}}_7 \rightarrow \widetilde{\textrm{C}}_7:$ Using $\hm{\Psi}, \hm{\Phi}$};
		
		\node [] at (2.5,-15) {\hspace{-46pt} \small$\bullet$ Adding ${\textrm{C}}_8:$ Using $\hm{\Phi}$};
	
		\draw[thick ,fill=white!10,rounded corners=0pt] (-7,-19) rectangle ++(14,3) node[midway]{ \shortstack[l]{\small \hspace{14pt}Problem in \eqref{maxmin4} w.r.t. \\    \small \hspace{-1pt} $[\lbrace {\mathbf{X}}_m\rbrace_{m=1}^{M},\mathbf{A}, \mathbf{Q},\hm{\Psi}, \hm{\Phi}, \theta]$}};	
		
		\draw[->,line width=1] (7,-17.5)--+(8,0);
		
		\draw[thick ,fill=white!10,rounded corners=0pt] (15,-19) rectangle ++(14,3) node[midway]{ \shortstack[l]{\small \hspace{17pt}Problem in \eqref{maxmin5} w.r.t.\\    \small \hspace{-2pt} $[\lbrace {\mathbf{X}}_m\rbrace_{m=1}^{M},\mathbf{A}, \mathbf{Q},\hm{\Psi}, \hm{\Phi}, \hm{\Omega}, \theta]$}};
		
		\node[draw=white!20, rotate=0, minimum height=-1 cm, minimum width=1 cm] at (11, -16.5) {\small Lemma~1};
		
		\node [] at (11,-20) {\hspace{0pt} \small$\bullet$ $\widetilde{\textrm{C}}_5 \rightarrow \bar{\textrm{C}}_5:$ Semi-definite relaxation};
		
		\node [] at (13.34-4.3,-21.5) {\hspace{06pt} \small$\bullet$ ${\textrm{C}}_3 \rightarrow \widetilde{\textrm{C}}_3,$ $\widetilde{\textrm{C}}_6 \rightarrow \bar{\textrm{C}}_6:$ MM};
		
		\node [] at (13.11,-23) {\hspace{-26pt} \small$\bullet$ ${\textrm{C}}_8 \rightarrow \widetilde{\textrm{C}}_8:$ S-procedure, introducing $\hm{\Omega}$};
		
		\node [] at (12.88-4.3,-24.5) {\hspace{-26pt} \small$\bullet$  $\widetilde{\textrm{C}}_8 \rightarrow \bar{\textrm{C}}_8:$ MM};
		
		\draw[->,line width=1] (17,-16)--+(0,5.5);

		\node [] at (24.5,-12.25) {\hspace{-43pt} \small$\bullet$ Adding ${\textrm{C}}_9 :$ Using $\hm{\Pi}$};
		
		\node [] at (24.5,-14.25) {\hspace{-28pt} \small$\bullet$ $\widetilde{\textrm{C}}_7 \rightarrow \bar{\textrm{C}}_7:$ Introducing $\hm{\Pi}$};

			\draw[thick ,fill=white!10,rounded corners=0pt] (15,-10.5) rectangle ++(14,3) node[midway]{ \shortstack[l]{\small \hspace{23pt}Problem in \eqref{maxmin9} w.r.t. \\    \small \hspace{-2pt} $[\lbrace {\mathbf{X}}_m\rbrace_{m=1}^{M},\mathbf{A}, \mathbf{Q},\hm{\Psi}, \hm{\Phi}, \hm{\Omega}, \hm{\Pi}, \theta]$}};
			
			\draw[->,line width=1] (17,-7.5)--+(0,3);
			
			\node [] at (24.5,-6) {\hspace{-67pt} \small$\bullet$ ${\textrm{C}}_9 \rightarrow \widetilde{\textrm{C}}_9:$ MM};

			\draw[thick ,fill=white!10,rounded corners=0pt] (15,-4.5) rectangle ++(14,3) node[midway]{ \shortstack[l]{\small \hspace{24pt}Problem in \eqref{maxmin6} w.r.t. \\    \small \hspace{-2pt} $[\lbrace {\mathbf{X}}_m\rbrace_{m=1}^{M},\mathbf{A}, \mathbf{Q},\hm{\Psi}, \hm{\Phi}, \hm{\Omega}, \hm{\Pi}, \theta]$}};

	\end{tikzpicture}
	\caption{The diagram of the subproblem in \ref{joint}.}
	\label{hh11}
	\centering
\end{figure}
Now, the second term in the objective function can be written as 
\begin{equation} \label{mm1}
\min_{1 \leq m \leq  {M}  } \hspace{1pt} \widetilde{\mathbf{x}}_m^H \widetilde{\hm{\Upsilon}}_{m}^{(\kappa-1)}\widetilde{\mathbf{x}}_m,
\end{equation}	
where ${\widetilde{\hm{\Upsilon}}}_{m}=\lambda_{m} \mathbf{I}_{\widetilde{N}} + \hm{\Upsilon}_{m}$ and $\lambda_{m}$ must be selected in such a way that ${\widetilde{\hm{\Upsilon}}}_{m}$ be a positive definite matrix \cite{soltanalian2014designing}.
Next, by defining a rank-1 matrix ${\mathbf{X}}_m=\widetilde{\mathbf{x}}_m \widetilde{\mathbf{x}}_m^H$, the quadratic term in \eqref{mm1} can be written as a linear term $\textrm{tr}\big( {\widetilde{\hm{\Upsilon}}}_{m}^{(\kappa-1)} {\mathbf{X}}_m \big)$ w.r.t. ${\mathbf{X}}_m$, and the constraint $\textrm{C}_5$ can be expressed as follows
\begin{equation} \label{ctil}
\widetilde{\textrm{C}}_{5}:\left[ {\mathbf{X}}_m \right]_{i,i}=p_m^{dl}[n],~ \textrm{rank} \left( {\mathbf{X}}_m \right)=1,~ {\mathbf{X}}_m\succeq \mathbf{0}, \hspace{2pt} \forall m,n,i.
\end{equation}
By using \eqref{ctil} and the linearized version of \eqref{mm1} as well as introducing an auxiliary variable $\theta$, the problem in \eqref{maxmin1} for fixed $[\widetilde{\mathbf{W}},L]$ can be reformulated as
\begin{align}\label{maxmin3}
&\max_{\lbrace {\mathbf{X}}_m\rbrace_{m=1}^{M},\mathbf{A}, \mathbf{P}^{ul}, \mathbf{Q}, \theta } ~~ \theta\\ \nonumber 
	&\mbox{s.t.} ~~\textrm{C}_{1}-\textrm{C}_{3}, \widetilde{\textrm{C}}_{5}, \textrm{C}_{6},
	\\ \nonumber &\scalemath{.96}{   \textrm{C}_{7}:   (1-\mu)  \min_{  \Delta \mathbf{r}^T_{k,i}  \Delta \mathbf{r}_{k,i} \leq \bar{d}^2_{k,i} } \hspace{-12pt} R_{k,i} +  \textrm{tr}\left( {\widetilde{\hm{\Upsilon}}}_{m}^{(\kappa-1)} {\mathbf{X}}_m \right) \geq \theta , \hspace{2pt} \forall k,i,m.}
\end{align}
%\begin{align}\label{maxmin3}
%	&\max_{\lbrace\widetilde{\mathbf{X}}_m\rbrace_{m=1}^{M},\mathbf{A}, \mathbf{P}^{ul}, \mathbf{Q}, \theta } ~ \theta\\ \nonumber 
%	&\mbox{s.t.} \\ \nonumber &\textrm{C}_{1},\textrm{C}_{2},\textrm{C}_{5}, \\ \nonumber &\textrm{C}_{4}:\textrm{tr} \left( \widetilde{\mathbf{X}} \right)=\widetilde{N}p^{dl}[n],~ \textrm{rank} \left( \widetilde{\mathbf{X}} \right)=1,~ \widetilde{\mathbf{X}}\succeq \mathbf{0}, \hspace{2pt} \forall n,  \\ \nonumber & \textrm{C}_{5}: \sum^{n}_{i=1} \tau_k[i] \delta_t  p^{{ul}}_k [i]  \leq E^0_k + \tau_0 \delta_t \epsilon_k \rho_0 \sum^{n}_{i=1}  \frac{p^{dl}[i]}{\phi_k [i] }, \hspace{2pt} \forall k,n,
%	\\ \nonumber &  \textrm{C}_{6}:  \textrm{tr} \left( {\widetilde{\hm{\Upsilon}}}^{(\kappa-1)} \widetilde{\mathbf{X}} \right) +(1-\mu) \frac{\delta_t}{N} \sum_{n=1}^{N}\tau_k[n]  \\ \nonumber & \hspace{20pt} \times \textrm{log}_2 \left(1+\frac{\rho_0 p^{{ul}}_k [n] }{{\sigma_c^2} [n] \phi_k [n] } \right) \geq \theta_b, \hspace{2pt} \forall k,
%	\\ \nonumber &  \textrm{C}_{7}:	{\parallel \mathbf{q} [n] - \mathbf{w}_k  \parallel }^2_2 +H^2 \leq \phi_k [n], \hspace{2pt} \forall k,n,
%\end{align}

Now, to proceed further, we focus on the first term in constraint $\textrm{C}_7$ and the right-hand side of the constraint $\textrm{C}_6$ which are associated with the circular uncertainty regions. 
By introducing $\phi_{k,m,i}[n]$ and applying a change of variable ${\psi}_{k,m}[n]=\tau_{k,m} [n] p^{ul}_{k,m}[n],\hspace{3pt} {\psi}_{k,m}[n] \geq 0$, the problem in \eqref{maxmin3} can be recast as  
\begin{align}\label{maxmin4}
	&\max_{\lbrace {\mathbf{X}}_m\rbrace_{m=1}^{M},\mathbf{A}, \mathbf{Q}, \hm{\Psi}, \hm{\Phi}, \theta } ~~ \theta\\ \nonumber 
	&\mbox{s.t.}~~ \textrm{C}_{1}-\textrm{C}_{3},\widetilde{\textrm{C}}_{5}, \\ \nonumber &	\widetilde{\textrm{C}}_{6}: \sum^{n}_{j=1} \delta_t \psi_{k,m} [j]  \leq   \tau_0 \delta_t \epsilon_{k,m} \rho_0 \sum^{n}_{j=1} \sum^{M}_{i=1}  \frac{p_i^{dl}[j]}{\phi_{k,m,i} [j] }  \\ \nonumber & \hspace{17pt}+ E^0_{k,m}, \hspace{3pt} \forall k,m,n, 
	\\ \nonumber &   \widetilde{\textrm{C}}_{7}:   (1-\mu) \frac{\delta_t}{N} \sum_{n=1}^{N} \tau_k [n]  \textrm{log}_2 \left(1+ \frac{\rho_0 \psi_{k,i} [n] }{ {\sigma_{c,i}^2} \tau_k [n] \phi_{k,i,i}[n] } \right)  \\ \nonumber & \hspace{18pt}+  \textrm{tr}\left( {\widetilde{\hm{\Upsilon}}}_{m}^{(\kappa-1)} {\mathbf{X}}_m \right) \geq \theta , \hspace{2pt} \psi_{k,i} [n] \geq 0, \hspace{2pt} \forall k,i,m,
		\\ \nonumber &  \textrm{C}_{8}:	{\parallel \widetilde{\mathbf{q}}_i [n] - ( \bar{\mathbf{r}}_{k,m} + \Delta \mathbf{r}_{k,m} ) \parallel }^2_2 +z_m^2 [n] \leq \phi_{k,m,i} [n], \\ \nonumber & \hspace{20pt} \Delta \mathbf{r}^T_{k,m}  \Delta \mathbf{r}_{k,m} \leq \bar{d}^2_{k,m}, \hspace{2pt} \forall k,m,i,n,
\end{align}
where 
\begin{equation}
\hm{\Psi}=\lbrace \psi_{k,m}[n],\forall k,m,n \rbrace,
\end{equation}
and
\begin{equation}
	\hm{\Phi}=\lbrace \phi_{k,m,i}[n],\forall k,m,i,n \rbrace.
\end{equation}

\begin{prop}\label{cor2}
The optimal solution of the problem in \eqref{maxmin3} can be obtained from the solution to the problem in \eqref{maxmin4}.
\end{prop}
\begin{proof}
Please refer to Appendix~\ref{app12}.
\end{proof}

It can be seen that the constraints $\textrm{C}_3$ and $\widetilde{\textrm{C}}_5$-$\widetilde{\textrm{C}}_7$ are still non-convex and $\textrm{C}_8$ has an infinite number of constraints due to the continuity of the corresponding user location uncertainty sets. The non-convexity of $\widetilde{\textrm{C}}_5$ originates from the rank-1 constraint. By adopting the semi-definite relaxation, we can drop the rank-1 constraint and proceed to solve the problem.

To proceed further, we can deal with the non-convexity of $\textrm{C}_3$ and $\widetilde{\textrm{C}}_6$ in the following. The left-hand side of both parts of $\textrm{C}_3$, along with the first term in the right-hand side of $\widetilde{\textrm{C}}_6$ can be minorized using their supporting hyperplane (see \cite{rezaei2019throughput,7898464} for more details). Therefore, the left-hand side of the first and second parts of $\textrm{C}_3$ as well as the first term in the right-hand side of $\widetilde{\textrm{C}}_6$ can be respectively obtained at the $\kappa$th iteration of the following expressions
\begin{align} \label{ttt1}
	{\parallel \widetilde{\mathbf{q}}^{(\kappa-1)}_m [n] - \mathbf{q}^{\mathrm{NFZ}}_{m,j} \parallel}^2_2 + 2 & \left( \widetilde{\mathbf{q}}^{(\kappa-1)}_m [n] - \mathbf{q}^{\mathrm{NFZ}}_{m,j} \right)^T \\ \nonumber & \times \left( \widetilde{\mathbf{q}}_m [n] - \widetilde{\mathbf{q}}^{(\kappa-1)}_m [n] \right),
\end{align}
\begin{align} \label{ttt12}
	\scalemath{.92}{-{\parallel \widetilde{\mathbf{q}}^{(\kappa-1)}_m [n] - \widetilde{\mathbf{q}}^{(\kappa-1)}_{m^{\prime}} [n] \parallel}^2_2 + 2} & \scalemath{.92}{\left(\widetilde{\mathbf{q}}^{(\kappa-1)}_m [n] - \widetilde{\mathbf{q}}^{(\kappa-1)}_{m^{\prime}} [n] \right)^T} \\ \nonumber & \scalemath{.92}{\times \left( \widetilde{\mathbf{q}}_m [n] - \widetilde{\mathbf{q}}_{m^{\prime}} [n] \right),}
\end{align}
\begin{align} \label{taumami}	
	\tau_0 \delta_t \epsilon_{k,m} \rho_0 \sum^{n}_{j=1}\sum^{M}_{i=1} \bigg \lbrace & \frac{p_i^{dl}[j]}{\phi^{(\kappa-1)}_{k,m,i} [j] } \\ \nonumber & - \frac{p_i^{dl}[j]\left( \phi_{k,m,i} [j] -\phi^{(\kappa-1)}_{k,m,i} [j] \right)}{\left(\phi^{(\kappa-1)}_{k,m,i} [j]\right)^2 }   \bigg \rbrace.   
\end{align}

Next, we consider the constraint $\textrm{C}_8$. Let us introduce a lemma which can be used to transform $\textrm{C}_8$ into a finite number of linear matrix inequalities (LMI)s.
\begin{lemma}[S-procedure]
Let a function $h_m (\mathbf{x}), m \in \lbrace 1,2 \rbrace, \mathbf{x} \in \complexC^{N}$, be defined as
\begin{equation}
	h_m (\mathbf{x})= \mathbf{x}^H   \mathbf{B}_m \mathbf{x} + 2 \Re \lbrace \mathbf{b}^{H}_m  \mathbf{x} \rbrace + b_m,
\end{equation}
where, $\mathbf{B}_m \in \mathbb{H}^{N \times N}, \mathbf{b}_m  \in \complexC^{N}$ and $b_m \in \realR$. Then, the implication $h_1 (\mathbf{x}) \leq 0  \Rightarrow h_2 (\mathbf{x}) \leq 0$ holds if
and only if there exists an $\omega \geq 0$ such that 
\begin{equation}
	\omega
	\begin{bmatrix}
		\mathbf{B}_1 & \mathbf{b}_1  \\  \mathbf{b}^{H}_1 &  {b}_1
	\end{bmatrix}
	- 
	\begin{bmatrix}
		\mathbf{B}_2 & \mathbf{b}_2  \\  \mathbf{b}^{H}_2 &  {b}_2
	\end{bmatrix} \succeq \mathbf{0},
\end{equation}
provided that there exists a point $\widehat{\mathbf{x}}$ such that $h_m(\widehat{\mathbf{x}}) < 0$. 
\end{lemma}
\begin{proof}Please see \cite{boyd2004convex}.\end{proof}
Then, we can rewrite constraint ${\textrm{C}}_{8}$ as 
\begin{align*}
	&\textrm{C}_{8}:  \Delta \mathbf{r}^T_{k,m}  \Delta \mathbf{r}_{k,m}  +2 \Re \lbrace   
	(  \bar{\mathbf{r}}_{k,m} - \widetilde{\mathbf{q}}_i [n])^T
	\Delta \mathbf{r}_{k,m} \rbrace \\ &+  (\widetilde{\mathbf{q}}_i [n] - \bar{\mathbf{r}}_{k,m} )^T (\widetilde{\mathbf{q}}_i [n] - \bar{\mathbf{r}}_{k,m} ) + z_m^2 [n]-  {\phi}_{k,m,i} [n] \leq 0.
\end{align*}
Then, using Lemma~1 and considering user location uncertianty region $\Delta \mathbf{r}^T_{k,m}  \Delta \mathbf{r}_{k,m} \leq \bar{d}^2_{k,m}$, we can equivalently rewrite the constraint ${\textrm{C}}_8$ as:
\begin{align} \label{mj10}
&	\widetilde{\textrm{C}}_8 : \mathbf{S} (\mathbf{Q},\hm{\Phi},\omega_{k,m,i}[n] )=\\ \nonumber & 
	\begin{bmatrix}
	\scalemath{.98}{	(\omega_{k,m,i}[n] -1) \mathbf{I}_{2}} & \scalemath{.98}{\widetilde{\mathbf{q}}_i [n] -\bar{\mathbf{r}}_{k,m}}  \\ 
		\scalemath{.98}{( \widetilde{\mathbf{q}}_i [n] -\bar{\mathbf{r}}_{k,m} )^T} & \substack{ \scalemath{.98}{- \omega_{k,m,i}[n] \bar{d}^2_{k,m}  + {\phi}_{k,m,i} [n]}
			\\   \scalemath{.98}{-   
				{\parallel \widetilde{\mathbf{q}}_i [n]- \bar{\mathbf{r}}_{k,m}\parallel }^2_2-z_m^2 [n]}}   
	\end{bmatrix}
	\succeq \mathbf{0},
\end{align}
with the variable $\omega_{k,m,i}[n] \geq0$. Note that the constraint $\widetilde{\textrm{C}}_8$ is still non-convex due to the quadratic term ${\parallel \widetilde{\mathbf{q}}_i [n]- \bar{\mathbf{r}}_{k,m}\parallel }^2_2$. For handling this, in light of MM, we construct a global underestimator for the mentioned quadratic term to minorize it and rewrite the constraint $\widetilde{\textrm{C}}_8$ at the $\kappa$th iteration as
\begin{align} \label{mj1}
	&	\bar{\textrm{C}}_8 : \mathbf{S} (\mathbf{Q},\hm{\Phi},\omega_{k,m,i}[n] )=\\ \nonumber & 
	\begin{bmatrix}
		\scalemath{.98}{	(\omega_{k,m,i}[n] -1) \mathbf{I}_{2}} & \scalemath{.98}{\widetilde{\mathbf{q}}_i [n] -\bar{\mathbf{r}}_{k,m}}  \\ 
		\scalemath{.98}{( \widetilde{\mathbf{q}}_i [n] -\bar{\mathbf{r}}_{k,m} )^T} & \substack{ \scalemath{.98}{- \omega_{k,m,i}[n] \bar{d}^2_{k,m}  + {\phi}_{k,m,i} [n]}
			\\   \scalemath{.98}{- \chi^{(\kappa)}_{k,m,i}[n]  
				-z_m^2 [n]}}   
	\end{bmatrix}
	\succeq \mathbf{0},
\end{align}
where 
\begin{align}
&\chi^{(\kappa)}_{k,m,i}[n] =	{\parallel \widetilde{\mathbf{q}}^{(\kappa-1)}_i [n] - \bar{\mathbf{r}}_{k,m} \parallel}^2_2 \\ \nonumber & + 2  \left( \widetilde{\mathbf{q}}^{(\kappa-1)}_i [n] - \bar{\mathbf{r}}_{k,m} \right)^T  \left( \widetilde{\mathbf{q}}_i [n] - \widetilde{\mathbf{q}}^{(\kappa-1)}_i [n] \right), \hspace{3pt}\forall k,m,i,n.
\end{align}
Now, based on the expressions in \eqref{ttt1}, \eqref{taumami}, and \eqref{mj1}, the problem in \eqref{maxmin4} can be restated as the following
\begin{align}\label{maxmin5}
	&\max_{\lbrace {\mathbf{X}}_m\rbrace_{m=1}^{M},\mathbf{A}, \mathbf{Q}, \hm{\Psi}, \hm{\Phi}, \hm{\Omega}, \theta } ~~ \theta\\ \nonumber 
	&\mbox{s.t.}~~ \textrm{C}_{1},~\textrm{C}_{2},~\widetilde{\textrm{C}}_{7},\\ \nonumber&
	\widetilde{\textrm{C}}_{3}: \eqref{ttt1} \geq \left(r^{\mathrm{NFZ}}_{m,j}\right)^2, \hspace{1pt} \forall m,n,j, ~ \eqref{ttt12} \geq d_{\textrm{min}}^2, \hspace{1pt} \forall m \neq m^{\prime},n, \\ \nonumber&
	\bar{\textrm{C}}_{5}:\left[ {\mathbf{X}}_m \right]_{i,i}=p_m^{dl}[n],~ {\mathbf{X}}_m\succeq \mathbf{0}, \hspace{2pt} \forall m,n,i,
	\\ \nonumber
	&\scalemath{.96}{	\bar{\textrm{C}}_{6}: \sum^{n}_{j=1} \delta_t \psi_{k,m} [j]  \leq  \eqref{taumami}+ E^0_{k,m}, ~ \forall k,m,n,}  
	\\ \nonumber &  \scalemath{.97}{ \bar{\textrm{C}}_{8}:	\mathbf{S} (\mathbf{Q},\hm{\Phi},\hm{\Omega} ) \succeq \mathbf{0},~\omega_{k,m,i}[n] \geq0, ~\forall k,m,i,n,}
\end{align}
where 
\begin{eqnarray}
	\hm{\Omega}=\lbrace \omega_{k,m,i}[n],\forall k,m,i,n \rbrace.
\end{eqnarray}
The logarithmic function of the first term in left-hand side of $\widetilde{\textrm{C}}_7$ is a non-concave term and so $\widetilde{\textrm{C}}_7$ is non-convex. By introducing the auxiliary variables $\pi_{k,i}[n]$, the problem in \eqref{maxmin5} can be equivalently rewritten as
\begin{align}\label{maxmin9}
	&\max_{\lbrace {\mathbf{X}}_m\rbrace_{m=1}^{M},\mathbf{A}, \mathbf{Q}, \hm{\Psi}, \hm{\Phi}, \hm{\Omega},\hm{\Pi}, \theta } ~~ \theta\\ \nonumber 
	&\mbox{s.t.}~~ \textrm{C}_{1},~\textrm{C}_{2},~\widetilde{\textrm{C}}_{3},~
	\bar{\textrm{C}}_{5},~	\bar{\textrm{C}}_{6},~\bar{\textrm{C}}_{8},
	\\ \nonumber &   \bar{\textrm{C}}_{7}:   (1-\mu) \frac{\delta_t}{N} \sum_{n=1}^{N} \tau_k [n]  \textrm{log}_2 \left(1+ \frac{\rho_0 \pi_{k,i}[n]  }{ {\sigma_{c,i}^2}  \tau_k [n]  } \right)  \\ \nonumber & \hspace{18pt}+  \textrm{tr}\left( {\widetilde{\hm{\Upsilon}}}_{m}^{(\kappa-1)} {\mathbf{X}}_m \right) \geq \theta , \hspace{2pt} \psi_{k,i} [n] \geq 0, \hspace{2pt} \forall k,i,m,
	\\ \nonumber &
	\textrm{C}_{9}:	\pi_{k,i}[n] \leq \frac{\psi_{k,i}[n]}{\phi_{k,i,i}[n]},~\forall k,i,n,
\end{align}
where 
\begin{eqnarray}
\hm{\Pi}=\lbrace \pi_{k,i}[n],\forall k,i,n \rbrace.
\end{eqnarray}
\begin{prop}\label{cor3}
The optimal solution of the problem in \eqref{maxmin5} can be obtained by solving the problem in \eqref{maxmin9}.
\end{prop}
\begin{proof}
Please refer to Appendix~\ref{app12}.
\end{proof}
Then, the non-convexity of $\textrm{C}_9$ can be dealt with by  minorizing its right-hand side using the following minorizer:
\begin{align} \label{keyjoint}
&\frac{\psi^{(\kappa-1)}_{k,i}[n]}{\phi^{(\kappa-1)}_{k,i,i}[n]} +\frac{1}{\phi^{(\kappa-1)}_{k,i,i}[n]} \left( \psi_{k,i}[n]-\psi^{(\kappa-1)}_{k,i}[n] \right) \\ \nonumber &- \frac{\psi^{(\kappa-1)}_{k,i}[n]}{\left(\phi^{(\kappa-1)}_{k,i,i}[n]\right)^2} \left( \phi_{k,i,i}[n]-\phi^{(\kappa-1)}_{k,i,i}[n] \right),
\end{align}
and the problem in \eqref{maxmin9} can be reformulated as
\begin{align}\label{maxmin6}
	&\max_{\lbrace {\mathbf{X}}_m\rbrace_{m=1}^{M},\mathbf{A}, \mathbf{Q}, \hm{\Psi}, \hm{\Phi}, \hm{\Omega},\hm{\Pi}, \theta } ~~ \theta\\ \nonumber 
	&\mbox{s.t.}~~ \textrm{C}_{1},~\textrm{C}_{2},~\widetilde{\textrm{C}}_{3},~\bar{\textrm{C}}_{5}-\bar{\textrm{C}}_{8}, ~
	\widetilde{\textrm{C}}_{9}:	\pi_{k,i}[n] \leq \eqref{keyjoint},~\forall k,i,n.
\end{align}
Since the logarithmic term in $\bar{\textrm{C}}_7$ is jointly concave w.r.t. $\tau_k [n]$ and $\pi_{k,i} [n]$, the constraint $\bar{\textrm{C}}_7$ and therefore, the problem in \eqref{maxmin6} are convex and can be solved efficiently by e.g., interior point methods. Note that $\mathbf{P}^{ul}$ can be synthesized after convergence of \eqref{maxmin6} as follows:
\begin{equation} \label{key2}
p_{k,m}^{ul}[n]=
\begin{cases}
\frac{\psi_{k,m}[n]}{\tau_{k,m} [n]},~  \tau_{k,m} [n] \neq 0,
\\
0, \hspace{29pt}  \tau_{k,m} [n]=0.
\end{cases}
\end{equation}
\subsection{Maximization over $L$ for fixed $[\widetilde{\mathbf{X}},\widetilde{\mathbf{W}}, \mathbf{A}, \mathbf{P}^{ul},\mathbf{Q}]$}\label{L}
As a final step, the problem in \eqref{maxmin1} w.r.t. the scalar $L$ can be solved via 1D search over its finite possible values in $\textrm{C}_4$.
\begin{algorithm}[t] \label{tt}
	\caption{The Proposed Method for Joint Maximization of Minimum Radar SINR and Minimum Communication Throughput}
	\begin{algorithmic}[tb]
		\STATE{\!\!\!\!\!\!\!\!\!\!\!\!\!  Main-0:}  Initialize ${\mathbf{X}}^{(i)}$, $L^{(i)}$, and set $i\leftarrow0$.
		\REPEAT
		\STATE{A:} Compute $\widetilde{\mathbf{w}}^{(i)}_m,~\forall m$ via the closed-form solution in \eqref{jjk}.
		\STATE{B-0:} Initialize $\widetilde{\mathbf{W}}^{(\kappa)}$, $L^{(\kappa)}$, $\hm{\Phi}^{(\kappa)}$, $\hm{\Psi}^{(\kappa)}$, and set $\kappa \leftarrow 0$.
		\REPEAT
		\STATE{B-1:} Solve the convex problem in \eqref{maxmin6}.
		\STATE{B-2:} Update $\kappa \leftarrow \kappa+1$.
		\UNTIL convergence
		\STATE{B-3:} Synthesize $\mathbf{P}^{ul}$ via \eqref{key2}.
		\STATE{C:} Solve the problem in \eqref{maxmin1} w.r.t. the scalar $L$ via 1D search over its finite possible values in $\textrm{C}_4$.
		\STATE{Main-1:} Update $i \leftarrow i+1$.
		\UNTIL convergence
		\STATE{\!\!\!\!\!\!\!\!\!\!\!\!\!  Main-2:} Synthesize $\widetilde{\mathbf{x}}_m$ form ${\mathbf{X}}_m$.
	\end{algorithmic}
\end{algorithm}
\subsection{Waveform Synthesis, Convergence, and Complexity Analysis}
Algorithm~1 summarizes the steps of the proposed method for jointly maximizing the minimum radar SINR and minimum communication throughput in a multi-UAV enabled integrated sensing and wireless powered communication system. The proposed method consists of
outer iterations which are denoted by superscript $i$. At each outer iteration, we have 3 steps associated with the subproblems in \ref{w}, \ref{joint} (which is denoted by superscript $\kappa$), and \ref{L}. At the end of the algorithm, we may synthesize the waveform $\mathbf{\widetilde{x}}_m$ from matrix $\mathbf{{X}}_m$ using e.g., the rank-1 approximation methods based on randomization techniques (see \cite{de2011design} for details).

Note that to ensure convergence to a stationary point, the sequence of objective values of the problem in \eqref{maxmin1} must be ascending in each subproblem. More precisely, let $g (.)$ denote the aforementioned objective function. We can write
\begin{align*}
	& \scalemath{.95}{g \left( {\widetilde{\mathbf{W}}}^{(i +1)}, [\widetilde{\mathbf{X}}^{(i +1)}, \mathbf{A}^{(i +1)}, \left(\mathbf{P}^{ul}\right)^{(i +1)},\mathbf{Q}^{(i +1)}],  L^{(i +1 )}\right)  \geq} \\ \nonumber & \scalemath{.95}{g \left( {\widetilde{\mathbf{W}}}^{(i +1)}, [\widetilde{\mathbf{X}}^{(i +1)}, \mathbf{A}^{(i +1)}, \left(\mathbf{P}^{ul}\right)^{(i +1)},\mathbf{Q}^{(i +1)}],  L^{(i )}\right) \geq} \\ \nonumber & \scalemath{.95}{g \left( {\widetilde{\mathbf{W}}}^{(i +1)}, [\widetilde{\mathbf{X}}^{(i)}, \mathbf{A}^{(i)}, \left(\mathbf{P}^{ul}\right)^{(i)},\mathbf{Q}^{(i)}],  L^{(i )}\right) \geq} \\ \nonumber & \scalemath{.95}{g \left( {\widetilde{\mathbf{W}}}^{(i)}, [\widetilde{\mathbf{X}}^{(i)}, \mathbf{A}^{(i)}, \left(\mathbf{P}^{ul}\right)^{(i)},\mathbf{Q}^{(i)}],  L^{(i)}\right) } ,
\end{align*}
where the inequalities above hold due to maximization performed in Subproblems \ref{w}, \ref{joint}, and \ref{L}, respectively. For the subproblems in \ref{w} and \ref{L}, the global maximum is obtained. Also, applying the proposed fractional programming and MM techniques to the design problem in \ref{joint} increases the associated objective function and, under mild conditions, provides stationary points of the problem \cite{rezaei2019throughput}. Therefore, due to boundedness of the objective function in \eqref{maxmin1}, the sequence of objective values in \eqref{maxmin1} obtained by the proposed method converges.

Next, the computational complexity of the proposed method is considered. For the subproblem in Subsection~\ref{w}, the closed-form expression in \eqref{jjk} must be calculated which needs matrix multiplication and inversion leading to the complexity of\footnote{This can be decreased to $\mathcal{O}(\widetilde{N}^{2.373})$ by using the optimized algorithms (see e.g., \cite{davie2013improved} for details).} $\mathcal{O}(\widetilde{N}^3)$ for $\mathrm{UAV}_m$. 
At each inner iterations of the subproblem in \ref{joint}, the dominant computational burden is associated with the constraints $\bar{\textrm{C}}_5$ and $\bar{\textrm{C}}_8$ due to adopting the semi-definite relaxation. Hence, considering the problem in \eqref{maxmin6}, the computational complexity is $\mathcal{O}(\sqrt{n}\textrm{log}(1/\epsilon)(mn^3 +m^2 n^2 +m^3))$ where $\epsilon > 0$ indicates the solution accuracy, $m=M(1+MK)$ is the number of semi-definite relaxation-based constraints, and $n=\widetilde{N}$ is associated with the size of the related positive semi-definite matrix \cite[Theorem~3.12]{bomze2010interior}. Finally, the 1D search in Subsection~\ref{L} can be performed via the complexity of $\mathcal{O}(M\widetilde{N}^2)$ which comes from the objective function calculations. For typical problem sizes, each iteration of the algorithm completes within seconds on a standard desktop machine. Moreover, since channel conditions may vary slowly, the joint optimization can be performed offline, allowing UAV trajectories to be pre-computed and updated as needed. This makes Algorithm 1 suitable for real-time or near real-time applications, particularly in scenarios where planning occurs over short to moderate time scales.
\begin{rema}
Note that the value of $\widetilde{{\mathrm{SINR}}}_m$ is greater than $R_{k,i}$ in the objective function of $\eqref{maxmin1}$ for a typical numerical setup (see Section~\ref{num}). Therefore, to preserve the controlling role of the Pareto weight $\mu$, we modify the objective function of \eqref{maxmin1} as follows
	\begin{align}\label{keyr}
	\scalemath{.96}{(1-\mu) \displaystyle \min_{\substack{{1 \leq k \leq  {K}}\\ {1 \leq i \leq  {M}}}} ~ \min_{  \Delta \mathbf{r}^T_{k,i}  \Delta \mathbf{r}_{k,i} \leq \bar{d}^2_{k,i} } R_{k,i}  + \mu \hspace{2pt}\mu_0 \displaystyle \min_{1 \leq m \leq  {M}  } \hspace{1pt} \widetilde{\mathrm{SINR}}_m,}
	\end{align}
where $\mu_0 \in (0,1]$	is a constant parameter which can be determined under the numerical supervision, without losing the optimality of the solution to the problem.
\end{rema}
%\begin{rema}
%It is worth pointing out that the proposed algorithm can be modified to address the sum throughput/sensing SINR maximization problem. The interested reader may follow the steps in Appendix~\ref{app1} for the sum utility problem.
%\end{rema}
 	\section{Non-linear EH} \label{nonl}
 In Section~\ref{sum}, for the sake of tractability in the problem formulation and the proposed solution, we simplified the process by assuming linear EH circuits. However, in practice, the RF-to-DC energy conversion in EH circuits exhibits a non-linear behavior\cite{rezaei2019throughput,10356111,rezaei2024cooperative }. Therefore, we consider a non-linear characteristic between the harvested energy, i.e., $E_{k,m} [n]$ and the received power, i.e., $ p_{k,m} [n]= \sum_{i=1}^{M} h_{k,m,i} [n] p_i^{dl}[n]$ at each ground user. To accommodate this non-linearity in the design problem formulated in \eqref{maxmin1}, we need to adapt constraint $\textrm{C}_{6}$. Initially, a modified version of \eqref{eee} for the harvested energy by $\mathrm{U}_{k,m}$ for non-linear model can be expressed as  \cite{boshkovska2017robust} 
 	\begin{equation}\label{hh}
 		E^{nl}_{k,m} [n]=\tau_0 \delta_t  \frac{\frac{\iota_{k,m}}{1+\textrm{exp} \left(-\widetilde{\iota}_{k,m} \left(  p_{k,m} [n]   -\bar{\iota}_{k,m} \right)  \right)}-\iota_{k,m} t_{k,m}}{1- t_{k,m}} ,
 	\end{equation}
 	\begin{equation*}
 		t_{k,m}=\frac{1}{1+\textrm{exp} \left(\widetilde{\iota}_{k,m} \bar{\iota}_{k,m} \right)} ,
 	\end{equation*}
 	where $\iota_{k,m}$, $\widetilde{\iota}_{k,m}$ and $\bar{\iota}_{k,m}$ are the curve fitting parameters (see  \cite{rezaei2019throughput} for more details). As a result, the constraint $\textrm{C}_{6}$ in \eqref{maxmin1} changes in such a way that $E_{k,m} [n]$ is replaced with $E^{nl}_{k,m} [n]$. Then, a similar process to the one in Subsection~\ref{joint} is followed, with the difference that in problem in \eqref{maxmin4}, the first term in the right-hand side of the constraint $\widetilde{\textrm{C}}_{6}$ changes as follows
 	\begin{equation*}\label{nnb}
 		\widetilde{ E}_{k,m} [n]=  \tau_0 \delta_t  \sum^{n}_{j=1}  
 		\frac{\frac{\iota_{k,m}}{1+\textrm{exp} \left(-\widetilde{\iota}_{k,m} \left(   \widetilde{p}_{k,m}[j]  -\bar{\iota}_{k,m} \right)  \right)}-\iota_{k,m} t_{k,m}}{1- t_{k,m}},
 	\end{equation*}
 	where
 	\begin{equation}
 		\widetilde{p}_{k,m}[j]= \rho_0  \sum^{M}_{i=1}  \frac{p_i^{dl}[j]}{\phi_{k,m,i} [j] }.
 	\end{equation}
 	Next, we must deal with the non-convexity of $\widetilde{\textrm{C}}_{6}$. Note that for typical values of $\iota_{k,m}$, $\widetilde{\iota}_{k,m}$ and $\bar{\iota}_{k,m}$, the term $\widetilde{ E}_{k,m} [n]$ is a non-decreasing concave function w.r.t. $\widetilde{p}_{k,m}[j]$. Also, due to the positivity of $\rho_0$, $p_i^{dl}[j]$ and $\phi_{k,m,i} [j]$ (see ${\textrm{C}}_{8}$ in \eqref{maxmin4}), the parameter $\widetilde{p}_{k,m}[j]$ is a convex function w.r.t. $\phi_{k,m,i} [j]$. Therefore, $\widetilde{ E}_{k,m} [n]$ is neither convex nor concave w.r.t. $\phi_{k,m,i} [j]$. In addressing this challenge, our approach involves initially establishing a suitably large parameter $\xi_{k,m}[n]$ with the condition $\nabla^2 \widetilde{E}_{k,m} [n] + \xi_{k,m}[n]\mathbf{I}_{Mn} \succeq \mathbf{0}$, where the Hessian is computed w.r.t. the vector $\hm{\phi}_{k,m}[n]   =\lbrace \phi_{k,m,i}[j],\forall k,m,i,j=1,...,n \rbrace$. Subsequently, we express $\widetilde{ E}_{k,m} [n]$ as the sum of a convex function and a concave one: 
 	\begin{align} \label{jj}
 		\widetilde{E}_{k,m} [n]=& \underbrace{  \widetilde{E}_{k,m} [n] + \frac{1}{2} \xi_{k,m}[n] {\hm{\phi}^{T}_{k,m}[n]} \hm{\phi}_{k,m}[n]}_{\text{convex}}
 		\\ \nonumber &	\underbrace{- \frac{1}{2} \xi_{k,m}[n] {\hm{\phi}^{T}_{k,m}[n]} \hm{\phi}_{k,m}[n]}_{\text{concave}}, ~\forall k,m,n.
 	\end{align}
 	Please see \cite[Appendix~A]{rezaei2019throughput} for a lower bound on $\xi_{k,m}[n]$.
 	Then, retaining the concave component, we can bound the convex part to yield a convex constraint in the following manner:
 	\begin{align} \label{e6}
 		\scalemath{.95}{\bar{\textrm{C}}_{6}:}& \scalemath{.95}{\sum^{n}_{j=1} \delta_t \psi_{k,m} [j]  \leq   \widetilde{E}^{(\kappa-1)}_{k,m}[n]  +  \frac{1}{2} \xi_{k,m}[n] \left({\hm{\phi}^{(\kappa-1)}_{k,m}[n]} \right)^T} \\ \nonumber & \scalemath{.95}{\times \hm{\phi}^{(\kappa-1)}_{k,m}[n]  +  {\mathbf{u}}^T_{k,m}[n]   \left( {\hm{\phi}_{k,m}[n]} - {\hm{\phi}^{(\kappa-1)}_{k,m}[n]} \right)} \\ \nonumber & \scalemath{.95}{-\frac{1}{2} \xi_{k,m}[n] \left({\hm{\phi}_{k,m}[n]} \right)^T \hm{\phi}_{k,m}[n]+ E^0_{k,m}, ~ \forall k,m,n,}
 	\end{align}
 	where
 	\begin{align*}
 		{\mathbf{u}}_{k,m}[n]=\Bigg [& \frac{\partial E^c_{k,m}[n] }{\partial \phi_{k,m,1}[1]}, \frac{\partial E^c_{k,m}[n] }{\partial \phi_{k,m,1}[2]},...,\frac{\partial E^c_{k,m}[n] }{\partial \phi_{k,m,1}[n]}, \\ \nonumber & \frac{\partial E^c_{k,m}[n] }{\partial \phi_{k,m,2}[1]},...,\frac{\partial E^c_{k,m}[n] }{\partial \phi_{k,m,M}[n]} \Bigg]^T,
 \end{align*}
and
 \begin{align*}
 &\scalemath{.95}{\frac{\partial E^c_{k,m}[n] }{\partial \phi_{k,m,i}[j]}= \xi_{k,m}[n] {\phi}^{(\kappa-1)}_{k,m,i}[j]+}\\ \nonumber & \scalemath{.92}{\frac{\tau_0 \delta_t \iota_{k,m}\widetilde{\iota}_{k,m} \rho_0 p^{dl}_i[j]}{(1-t_{k,m})\left( 1+\textrm{exp} \left(-\widetilde{\iota}_{k,m} \left(  \widetilde{p}^{(\kappa-1)}_{k,m}[j]    - \bar{\iota}_{k,m} \right) \right) \right)^2 \left({\phi}^{(\kappa-1)}_{k,m,i}[j]\right)^2},}
 \end{align*}
where $E^c_{k,m}[n]$ denotes the convex part of the expression in \eqref{jj}. Now, through substituting the constraint $\bar{\textrm{C}}_{6}$ in \eqref{e6} with the corresponding constraint $\bar{\textrm{C}}_{6}$ of the problem in \eqref{maxmin6}, Algorithm 1 can be adjusted to accommodate the non-linear EH paradigm.
 \section{Numerical Examples} \label{num}	
In this section, we evaluate the effectiveness of the proposed method by numerical examples. The convex problem associated with the devised method is solved by CVX \cite{cvx}. Considering the complexity and multitude of parameters involved in our proposed method, this section is divided into three subsections with specific assumptions to facilitate better numerical comprehension. It's important to highlight that the non-linear EH constraint examined in Section~\rom{4} are solely addressed in Subsection~\ref{bbn}. Shared parameters across all subsections are outlined below, while unique parameters are determined within each subsection. We consider $\mathrm{\widetilde{SIR}}_{m,k}=\frac{\vert \alpha_{m} \vert^2}{\widetilde{\sigma}^2_{m,k}}=-10$~dB, $\forall m,k$, and $\mathrm{\widetilde{SNR}}_{m,l}=\frac{\vert \alpha_{m} \vert^2}{{\sigma}^2_{m,l}}=-10$ dB, $\forall m,l$
%, and $\omega_d=\pi/4$ 
for radar receiver; $\widetilde{N}=350$, $\widetilde{\tau}_0= 700$ $\mu\hspace{1pt}$second, $lb=5$, $ub=10$, and $p_m^{dl}[n]=37$ dBm, $\forall m,n$ \cite{nguyen2022ris} for radar/WPT waveform; $\sigma_{c,m}^2=-134$ dBm, $\forall m$ \cite{nguyen2022ris} for communication receiver; $\nu_m=0.05$ radians (for implementing the FFT-based Doppler processing filter), $\forall m$  \cite{li2022joint}, $\zeta=30$ degrees, $v_{\textrm{max}}=20$ m/s \cite{li2022joint}, $d_{\textrm{min}}=5$\hspace{2pt}m, and $\delta_t=1$ second \cite{wei2022safeguarding} as UAV flight parameters; $\rho_0=-30$ dB \cite{li2022joint} for channel power gain; $\epsilon_{m,k}=0.5, \hspace{1pt} \forall m,k$ \cite{nguyen2022ris}, and $E^0_{m,k}=1$ mJ, $\forall m,k$ for EH circuit; and $\mathrm{CR}_m = \lbrace \mathrm{CR}^x_m= 300\hspace{1pt}\textrm{m} \times \mathrm{CR}^y_m=300\hspace{1pt}\textrm{m} \rbrace,$ $\forall m$ \cite{wei2022safeguarding} for cluster regions. Also, we assume $\mu=0.5$ unless otherwise specified. Moreover,
we define the normalized radius of user location uncertainty as
\begin{equation}
\widetilde{r}_{k,m}=\frac{\bar{d}_{k,m}}{\frac{\textrm{min}\hspace{2pt} (\mathrm{CR}^x_m, \mathrm{CR}^y_m)}{2}},~\forall k,m.
\end{equation}
\subsection{Simplified 2D flight scenario}\label{2dsimp} 
\begin{figure} 
	\centering
	\includegraphics[width=9.5cm,height=7cm]{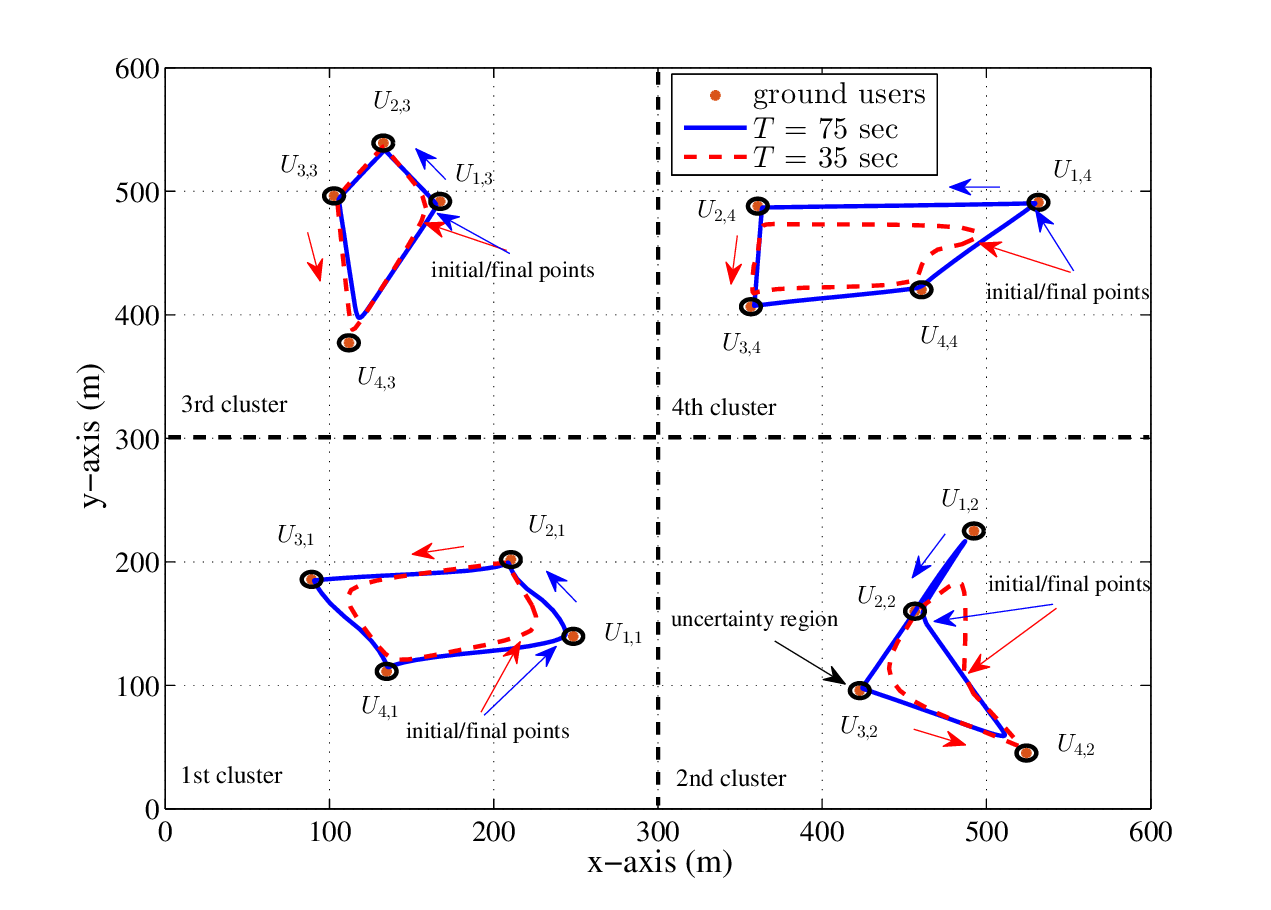}
	\caption{2D optimized UAV trajectories for different $T$.}
	\label{jhjh}
	\centering
\end{figure}
\begin{figure} 
	\centering
	\includegraphics[scale=.4]{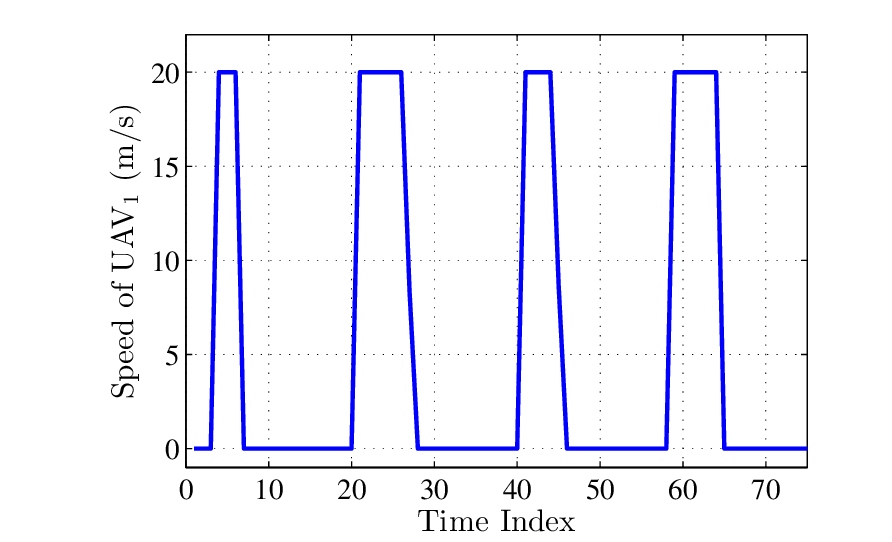}
	\caption{The speed of $\mathrm{UAV}_1$ for $T=75$ seconds.}
	\label{jhjh12}
	\centering
\end{figure}
First, we consider a 2D flight setup, i.e., $z_{\textrm{max}}=z_{\textrm{min}}=z^{\mathrm{tr}}_m= 100$ m, $\forall m$ \cite{li2022joint}, with $K=4$ for ground users which are located in $M=4$ clusters. Moreover, we set $N_{m}^{\mathrm{NFZ}}=0$, $\forall m$, and $\widetilde{r}_{k,m}=0.04$, $\forall k,m,$ in this subsection. 
%	Note that the setup parameters are chosen in such a way that the required radar field-of-view (FOV) is provided. However, one can consider a radar FOV constraint in the optimization problem to provide a mathematical guarantee on radar FOV (which is left for future works).

Fig.~\ref{jhjh} illustrates the optimized UAV trajectories for different $T$. It can be seen that in order to increase the harvested energy during the sensing/WPT phase, the UAVs adjust their trajectory center to be close to the center of users for all cases. By increasing $T$, the UAVs try to move closer to each user for increasing the communication throughput during the uplink phase. Precisely, UAVs hover around their cluster users for the maximum possible duration to maintain the closest situation. 
For instance, the amount of hovering time for $\mathrm{UAV}_1$ can be seen from its speed diagram in Fig.~\ref{jhjh12} for the case of $T=75$ seconds, where we can observe that the speed of $\mathrm{UAV}_1$ reduces to zero when flies right above each user. 
%Also, it is assumed that the radars have a $250\hspace{2pt}\textrm{m} \times 250\hspace{2pt}\textrm{m}$ field-of-view (FOV) at $H=100\hspace{2pt}\textrm{m}$ and therefore, can cover the entire area of their clusters during the period of $T$ sec.

Moreover, as an example for time scheduling, we illustrate the optimized fraction of the uplink time resource allocation of $\mathrm{UAV}_1$ for the case of $T=75$ seconds in Fig.~\ref{jhjh1}, where the optimal value of $L$ is equal to $512$ and therefore, $\tau_0=L \widetilde{\tau}_0=0.3584$ second is obtained for the sensing/WPT phase. We can also observe that in the optimized uplink time scheduling, only one user in each cluster (which is the closest to its associated UAV) is supported at each subslot.

In Fig.~\ref{jhjh2}, the Pareto curves along with the optimized value of $L$ are shown for different values of Pareto weight $\mu$ assuming $T=35$ seconds. It is observed that by increasing $\mu$ till $0.85$, the minimum sensing SINR is increasing; and minimum communication throughput is decreasing. This is due to the fact that more attention is given to the sensing SINR which is confirmed by looking at larger values for optimal $L$. For $\mu=0.85$, $L$ reaches its upper bound $1024$ and so, larger $\mu$ does not change time scheduling and Pareto curves. Note that since the maximum power budget
for uplink transmission are determined by the amount of harvested
energy (see \eqref{key}), the performance of recharging procedure directly affects
the communication throughput. Indeed, the throughput values indicate the performance of both WPT and WIT phases.

As a final note, to guide the optimal selection of the parameter $\mu$, it is important to recognize its direct relationship with the parameter $L$. Specifically, increasing $\mu$ from values close to 0 toward 1 results in $L$ increasing from its lower bound $2^{lb}$ to its upper bound $2^{ub}$. As discussed in Subsection~\ref{mmnh}, $L$ plays a critical role in balancing the trade-off between sensing and communication. From the radar perspective, $L$ influences multiple design aspects. For example, a higher value of $L$ allows for an increased minimum required sensing SINR at the radar receiver, which improves target detection performance. However, excessively large values of $L$ may lead to increased complexity in the radar processing chain (see Fig.~\ref{hh1}). From the communication perspective, while increasing $L$ can enhance WPT performance, it may simultaneously degrade WIT efficiency.
\begin{figure} 
	\centering
	\includegraphics[scale=.4]{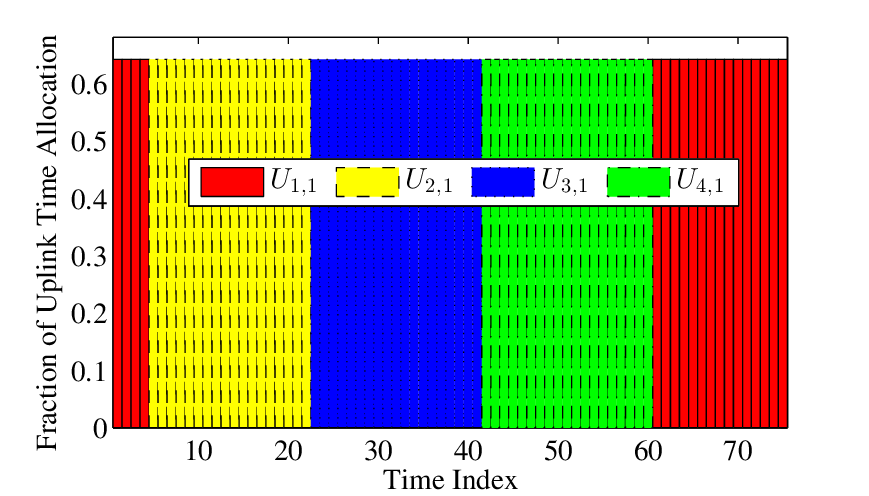}
	\caption{Optimized time resource allocation of $\mathrm{UAV}_1$ for $T=75$ seconds.}
	\label{jhjh1}
	\centering
\end{figure}
\begin{figure} 
	\centering
	\includegraphics[scale=.45]{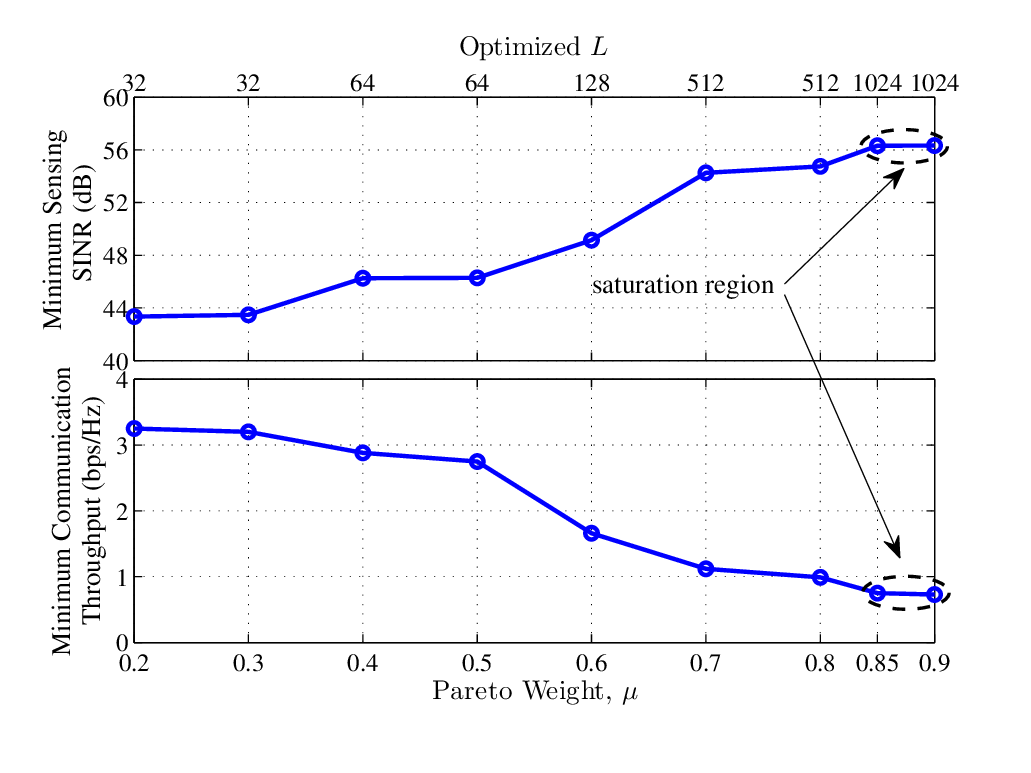}
	\caption{Pareto-optimized curves for $T=35$ seconds.}
	\label{jhjh2}
	\centering
\end{figure}
\subsection{Real-world 2D flight scenario}\label{bbn}
In this subsection, we aim to explore a 2D flight scenario within the constraints of real-world conditions. Specifically, we take into account the non-linearities inherent in EH circuits as detailed in Section~\ref{nonl}. Building upon the parameter settings outlined in Subsection~\ref{2dsimp}, we adjust the values of $M$ and $K$ to $M=1$ and $K=6$, respectively. Furthermore, we set $\iota_k=24$ mW, $\widetilde{\iota}_k=1500$ and $\bar{\iota}_k=22\times 10^{-4}, \hspace{1pt} \forall k$ \cite{boshkovska2017robust}.

Additionally, we present two baseline schemes for performance evaluation. The first baseline approach entails partial optimization, where all design variables are optimized except for the trajectory of UAVs. Specifically, UAVs follow a straight trajectory with a constant speed from their initial locations $\mathbf{q}_m [0]$ to their final locations $\mathbf{q}_m [N]$, calculated as ${\parallel \mathbf{q}_m [N]- \mathbf{q}_m [0] \parallel}_2 /T$ \cite{meng2022throughput}. Similarly, the second baseline approach involves partial optimization, where UAVs initially fly straight from their starting locations to an optimized point at maximum speed ${v}_{\textrm{max}}$. After hovering at this position, UAVs proceed to their final destinations in a straight path at maximum speed ${v}_{\textrm{max}}$ \cite{meng2022throughput}. It is worth noting that both baseline methods are special cases of the proposed approach discussed in Section~\ref{sum}. Furthermore, for a more comprehensive comparison with the baseline methods, the constraint $ \mathbf{q}_m [0] = \mathbf{q}_m[N] ,\hspace{2pt} \forall m,$ originally defined in constraint set $\textrm{C}_{2}$ of the design problem presented in \eqref{maxmin1}, is adjusted to $\mathbf{q}_m [0]= \mathbf{q}^I_m,~ \mathbf{q}_m[N]= \mathbf{q}^F_m $ within this subsection. Here, $\mathbf{q}^I$ and $\mathbf{q}^F$ denote the predefined initial and final locations of $\mathrm{UAV}_m$, respectively.

Fig.~\ref{jhjh100} presents the optimized UAV trajectories for both the proposed method and the baseline methods. As discussed in Subsection~\ref{2dsimp}, reducing the duration $T$ results in UAVs being unable to reach each user precisely and hover at their locations for the designated time. Increasing the parameter $K$ further highlights this observation in the optimized UAV trajectories under real-world conditions. Furthermore, in the case of the second baseline method, the optimal hovering point is determined to not only fulfill sensing requirements but also enhance the minimum communication throughput of users to the greatest extent possible. Additionally, Fig.~\ref{pow} depicts the impact of downlink transmission power $p_m^{dl}[n]$ on the objective function of the problem in \eqref{maxmin1}. It is evident that the proposed scheme exhibits superior performance compared to the baseline methods.
\begin{figure} 
	\centering
	\includegraphics[scale=.46]{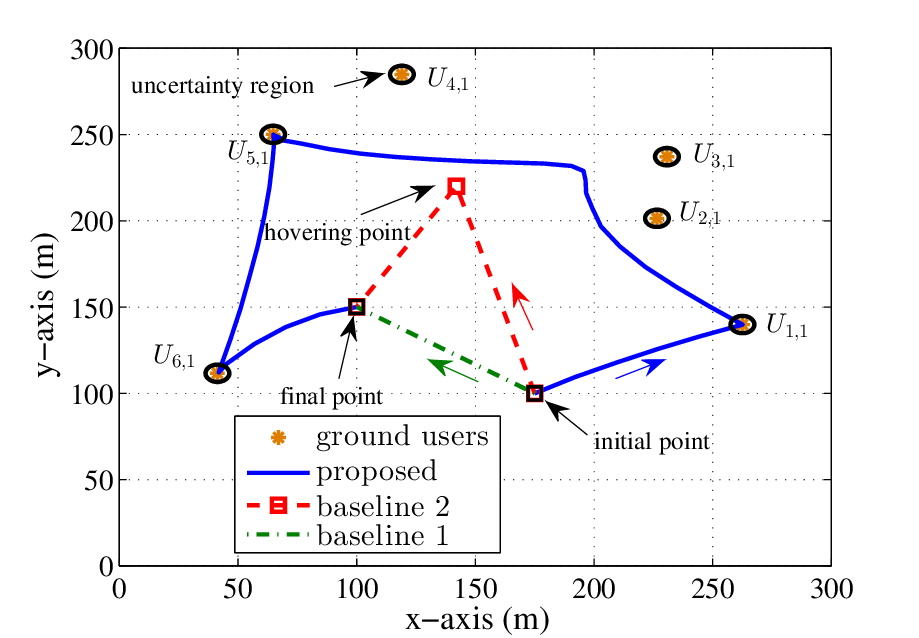}
	\caption{2D optimized UAV trajectory under real-world constraints for $T=50$ seconds.}
	\label{jhjh100}
	\centering
\end{figure}
\begin{figure} 
	\centering
	\includegraphics[scale=.46]{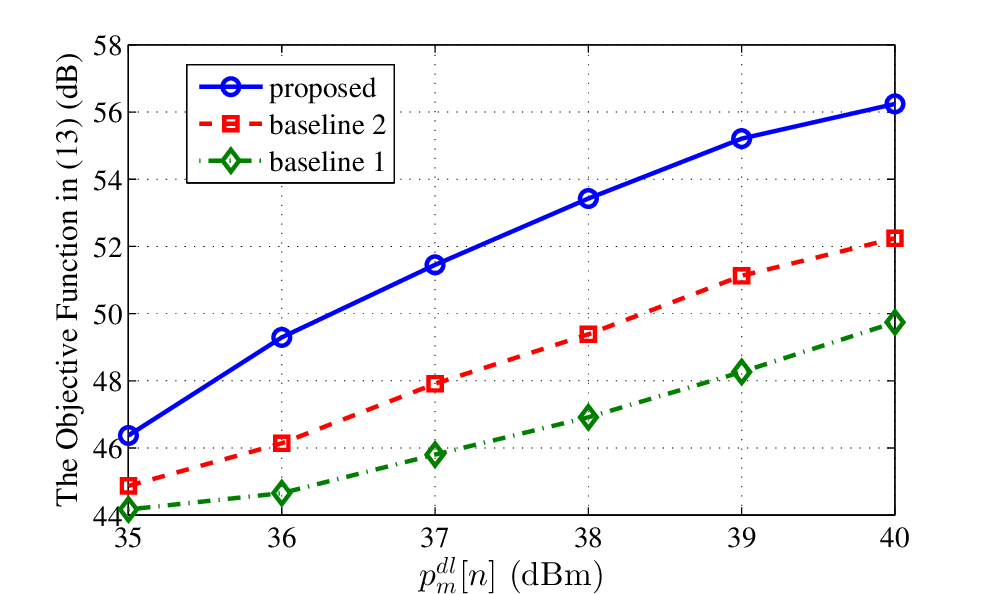}
	\caption{The values of the objective function in \eqref{maxmin1} versus the downlink transmission power $p^{dl}_m [n]$ for $T=50$ seconds.}
	\label{pow}
	\centering
\end{figure}

Next, we compare the proposed method, which incorporates a non-linear EH model, with benchmark schemes that assume a linear EH model, as considered in \cite{10382465, 10539920}. These linear EH-based methods are treated as a third baseline in this section. Fig.~\ref{pow12} illustrates the performance comparison between the proposed scheme and this baseline. Since the amount of harvested energy directly determines the uplink power budget and consequently impacts the communication throughput of users, we adopt the minimum communication throughput as the comparison metric. As described in \cite{rezaei2019throughput,rezaei2024cooperative}, the input-output power characteristic of the non-linear EH model used in Section~\ref{nonl} consists of three distinct regions: the linear region, a transient region, and a saturation region. As observed in Fig.~\ref{pow12}, when the downlink transmit power $p_m^{dl}[n]$ remains within the linear region of the EH circuit, the throughput performance of the proposed non-linear model is comparable to that of the linear baseline. However, as $p_m^{dl}[n]$ increases and the system enters the transition and saturation regions, the performance of the linear EH-based methods surpasses that of the non-linear model, due to reduced energy conversion efficiency. Therefore, Fig.~\ref{pow12} helps to identify the optimal operating region for the downlink transmit power $p_m^{dl}[n]$, which should ideally lie within the linear regime of the non-linear EH circuit to ensure favorable throughput performance.
\begin{figure} 
	\centering
	\includegraphics[scale=.46]{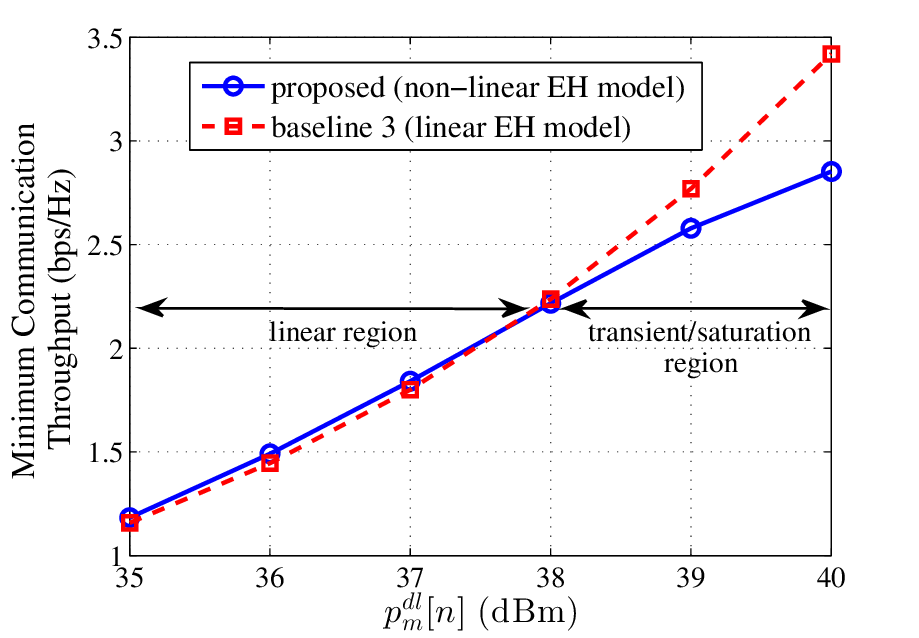}
	\caption{Minimum communication throughput versus downlink transmission power $p^{dl}_m [n]$ for non-linear and linear EH models.}
	\label{pow12}
	\centering
\end{figure}
\subsection{3D flight scenario}
In this subsection, we study the general 3D flight mode with $z_{\textrm{max}}=150$ m, $z_{\textrm{min}}= 50$ m, $K=5$, $M=2$, $N_1^{\mathrm{NFZ}}=1$, $N_2^{\mathrm{NFZ}}=2$,  $r_{m,j}^{\mathrm{NFZ}}=10$ m, $\forall m,j$, $z_m^{\mathrm{tr}}= 85$ m, $\forall m$, and $\widetilde{r}_{k,m}=0.03$, $\forall k,m$.

Fig.~\ref{jhjh125} shows the UAV trajectories in a 3D scenario. The ability of the proposed method to avoid collision with obstacles in the NFZ can be seen from this figure.

In Fig.~\ref{jhjh20}, we investigate the impact of user location uncertainty on the communication throughput by comparing the proposed robust scheme with a non-robust baseline, referred to as the fourth baseline method in this section. This baseline follows the approaches of works such as \cite{948449611, 808617712}, where the user locations are assumed to be perfectly known. More specifically, the robust scheme considers location uncertainties during the resource allocation stage, whereas the fourth baseline scheme operates under the assumption that user locations are precise and free from uncertainties, neglecting to consider location uncertainties during the design phase. As expected, by increasing the normalized radius of uncertainty, i.e., $\widetilde{r}_{k,m}$, the minimum communication throughput decreases. Also, it can be observed that the performance gain of the robust method over the non-robust one is significant for higher values of $\widetilde{r}_{k,m}$. This is because the non-robust scheme lacks readiness for managing uncertainties during the design stage, while the robust method acknowledges and attempts to mitigate uncertainties to the best of its ability.
\begin{figure}
	\centering
	\subfigure[]{\includegraphics[width=9.25cm,height=5.cm]{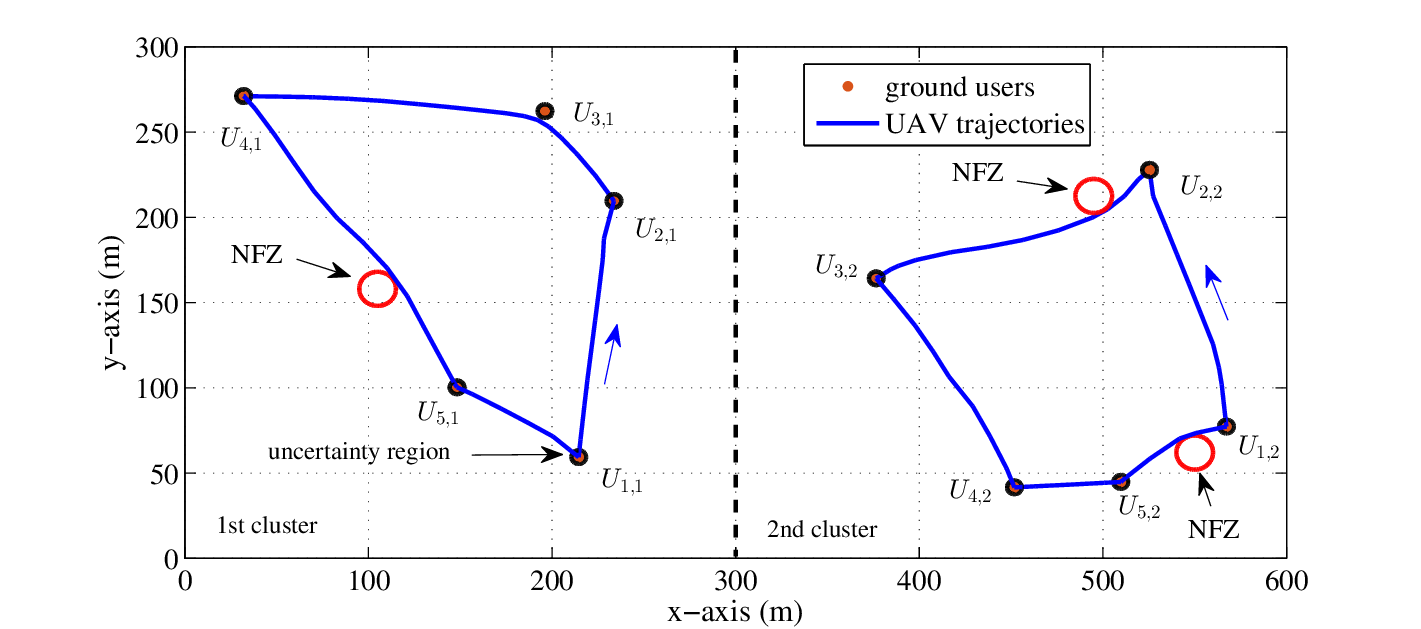}}
	\subfigure[]{\includegraphics[width=8.65cm,height=5.cm]{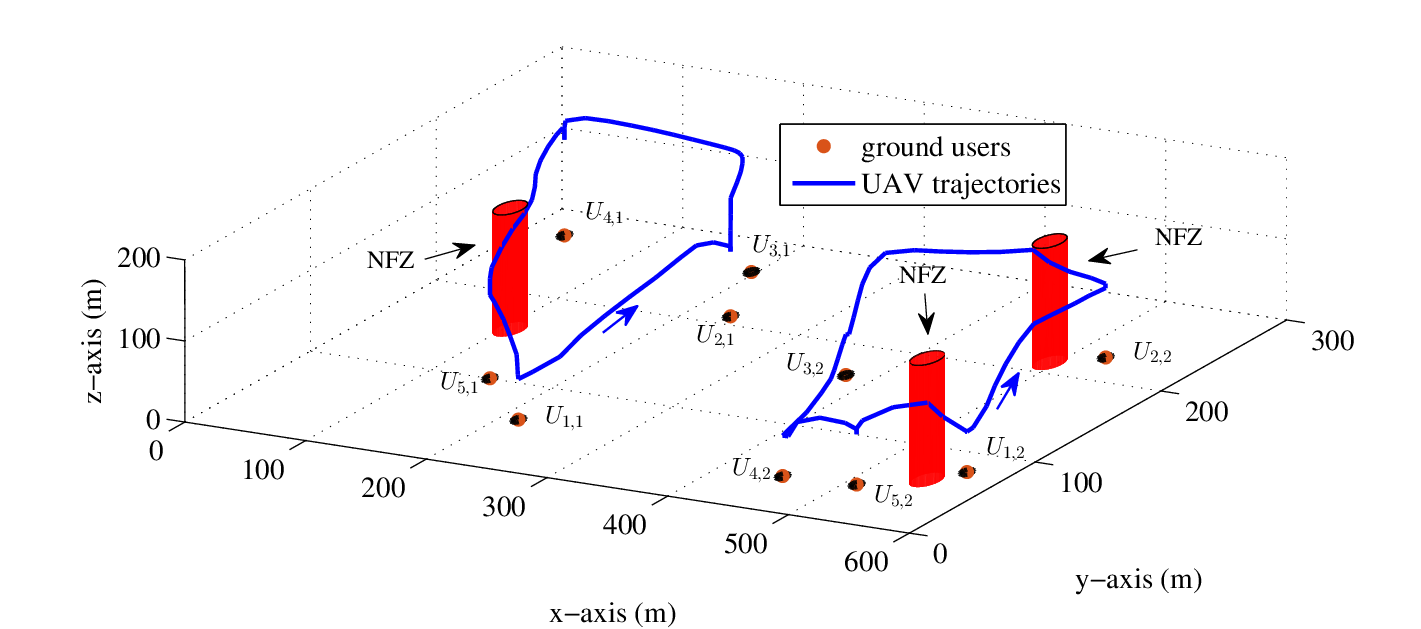}}
	\caption{3D optimized UAV trajectories for $T=60$ seconds: (a)~bird-view,~(b)~full-view.}
	\label{jhjh125}
\end{figure}	
\begin{figure} 
	\centering
	\includegraphics[scale=.46]{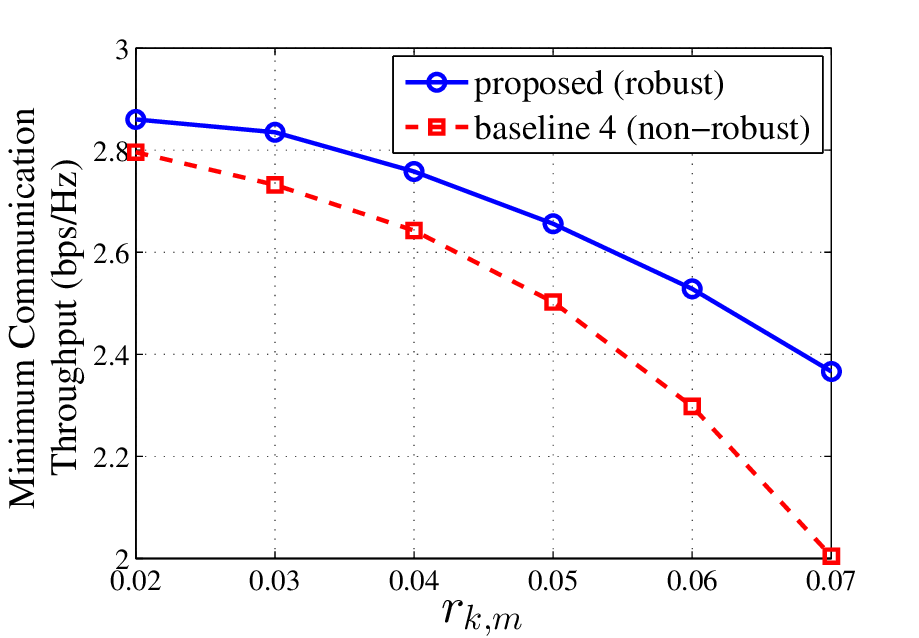}
	\caption{Minimum communication throughput versus the normalized user location uncertainty for $T=40$ seconds.}
	\label{jhjh20}
	\centering
\end{figure}

Finally, to highlight the importance of joint sensing and WPT waveform design, we introduce the fifth baseline method that employs a collaborative MRT-based waveform during the downlink phase \cite{moorthy2023swarm, 10574401}. Specifically, each single-antenna UAV transmits an energy-carrying signal with a phase aligned to the conjugate of its individual channel toward the designated user. This strategy allows all UAVs to cooperatively form a distributed MRT beam, thereby maximizing the coherent sum of received power at the user side. Table~\ref{XXXX} shows that, while collaborative MRT improves the harvested energy, it fails to satisfy the sensing SINR constraints—underscoring the effectiveness of the ISAC-optimized waveform.
\begin{table}[htbp]
	\centering
	\caption{Comparison of Proposed and Baseline 5 schemes.}\label{XXXX}
	\renewcommand\cellalign{cc}
	\renewcommand\theadalign{cc}
	\begin{tabular}{|c|c|c|c|c|}
		\hline
		\multirow{2}{*}{$p_m^{dl}[n]$} & \multicolumn{2}{c|}{\makecell[c]{\textbf{Minimum Communication}\\\textbf{Throughput (bps/Hz)}}} & \multicolumn{2}{c|}{\textbf{Sensing SINR (dB)}} \\
		\cline{2-5}
		& Proposed & Baseline 5 & Proposed & Baseline 5 \\
		\hline
		36 dBm & 1.53 & 1.68 & 48.31 & 43.22 \\
		\hline
		38 dBm & 2.48 & 2.71 & 52.13 & 48.03 \\
		\hline
		40 dBm & 3.10 & 3.18 & 55.75 & 51.12 \\
		\hline
	\end{tabular}
\end{table}
\section{Conclusion} \label{con}
In this paper, we proposed a multi-UAV enabled integrated sensing and wireless powered communication framework where a dual use of radar/WPT waveforms enables the UAVs to efficiently detect targets and serve a group of energy-limited ground users. We designed the radar receive filters, radar/WPT waveforms, uplink power along with time scheduling of ground users, and UAV trajectories to maximize a joint radar and communication performance metric under user location uncertainty. Through simulations, we demonstrated that the proposed method improves the joint performance of the radar and wireless powered communication, while choosing a suitable value for the total length of radar/WPT sequences. The proposed model in this paper can be extended to a class of ISAC optimization problems with other performance metrics under some general constraints, some of which are discussed as follows for
future work.
\begin{itemize}
	\item One can consider the mutual information as the radar performance metric instead of minimum SINR in the cost function of the proposed multi-objective design problem. 
%	\item Sum utility maximization, i.e., sum of sensing SINR and communication throughput can be considered as the objective function in the design problem. Note that to deal with the sum sensing SINR maximization, the  some of non-convex fractional terms must be maximized which is an interesting mathematical challenge.
	\item 
	Some practical concepts such as UAV jittering and UPA of antennas for all sensing/communication transmitters and receivers can be taken into account in the design problem. 
\end{itemize}
\appendices 
\section{Proof of Proposition~\ref{cor1}} \label{app11}
%Let us rewrite the objective function of the subproblem in \ref{joint} w.r.t. $\mathbf{\widetilde{X}}$ as 
%\begin{equation}
%q(\mathbf{\widetilde{X}})= a + \displaystyle \mu \min_{1\leq m \leq M} f_1(\mathbf{\widetilde{x}}_m),
%\end{equation}
%where
%\begin{equation}
%	a= (1-\mu) \displaystyle \min_{\substack{{1 \leq k \leq  {K}}\\ {1 \leq i \leq  {M}}}} ~ \min_{  \Delta \mathbf{r}^T_{k,i}  \Delta \mathbf{r}_{k,i} \leq \bar{d}^2_{k,i} } R_{k,i},
%\end{equation}
%is a constant w.r.t. $\mathbf{\widetilde{X}}$ and $f_1(\mathbf{\widetilde{x}}_m)$ is defined in Subsection~\ref{joint}. 
Suppose that $f_2(\widetilde{\mathbf{x}}_m)= g_1(\widetilde{\mathbf{x}}_m) - f_1(\widetilde{\mathbf{x}}_m^{(0)}) g_2(\widetilde{\mathbf{x}}_m),\hspace{2pt}\forall m$ where $\widetilde{\mathbf{x}}_m^{(0)}$ indicates the current value of $\widetilde{\mathbf{x}}_m$. Now, let us define $\widetilde{\mathbf{x}}_m^{\star}= \textrm{arg} \max_{\widetilde{\mathbf{x}}_m} f_2(\widetilde{\mathbf{x}}_m), \hspace{2pt}\forall m$. It is observed that $f_2(\widetilde{\mathbf{x}}^{\star}_m) \geq f_2(\widetilde{\mathbf{x}}^{(0)}_m)=0, \hspace{2pt}\forall m$. As a result, since $g_2(\widetilde{\mathbf{x}}_m) >0$, $f_2(\widetilde{\mathbf{x}}^{\star}_m)= g_1(\widetilde{\mathbf{x}}^{\star}_m) - f_1(\widetilde{\mathbf{x}}_m^{(0)}) g_2(\widetilde{\mathbf{x}}^{\star}_m) \geq 0,\hspace{2pt}\forall m$ leads to $f_1(\widetilde{\mathbf{x}}^{\star}_m) \geq f_1(\widetilde{\mathbf{x}}^{(0)}_m),\hspace{2pt}\forall m $, and therefore, $q(\widetilde{\mathbf{X}}^{\star}) \geq q(\widetilde{\mathbf{X}}^{(0)})$, where $\widetilde{\mathbf{X}}^{\star}=\lbrace \widetilde{\mathbf{x}}^{\star}_m ,\hspace{2pt} \forall m \rbrace$ and $\widetilde{\mathbf{X}}^{(0)}=\lbrace \widetilde{\mathbf{x}}^{(0)}_m ,\hspace{2pt} \forall m \rbrace$. Consequently, $\widetilde{\mathbf{X}}^{\star}$ can be a new matrix $\widetilde{\mathbf{X}}$
which increases $q(\widetilde{\mathbf{X}})$.
\section{Proof of Proposition~\ref{cor2} and Proposition~\ref{cor3}} \label{app12}
For the case of Proposition~\ref{cor2}, we need to prove the following expressions:  
\begin{itemize}
	\item at the optimal solution of the problem in \eqref{maxmin4}, $\textrm{C}_8$ is active;
	\item the optimal value of the problem in \eqref{maxmin3}, denoted by $\theta_{\ref{maxmin3}}$, is always greater than the optimal value of the problem in \eqref{maxmin4}, denoted by $\theta_{\ref{maxmin4}}$, and the equality holds when $\textrm{C}_8$ is active.
\end{itemize}
Let us proceed by contradiction to prove the first item. To this end, let us assume
that at the optimal point of the problem in \eqref{maxmin4}, the equality in $\textrm{C}_8$ does not hold. In this case, the value of $\theta$ can be increased by reducing $\phi_{k,m,i}[n]$ in $\textrm{C}_8$, which is evidently in contradiction with the assumption of $\textrm{C}_8$ is not active.

As to the second item, considering $\textrm{C}_7$ in \eqref{maxmin3} and $\widetilde{\textrm{C}}_8$ as well as $\textrm{C}_7$ in \eqref{maxmin4}, straightforwardly leads to the fact that $\theta_{\ref{maxmin3}} \geq \theta_{\ref{maxmin4}}$ and the equality holds holds when $\textrm{C}_8$ is active. 

The Proposition~\ref{cor3} can be similarly proved.
\bibliographystyle{IEEETran}
\bibliography{myreff}
\end{document}